\documentclass[a4paper,12pt,reqno]{amsart}

\usepackage{a4}
\usepackage{verbatim}
\usepackage{comment}
\usepackage{color}
\usepackage{todonotes}
\usepackage{doi}

\usepackage{amssymb,stmaryrd}
\usepackage{color}
\usepackage{graphicx}
\usepackage{floatflt}
\usepackage{float} 		

\usepackage{amsthm}
\usepackage{amscd}
\usepackage{hyperref}

\newtheorem{definition}{Definition}

\newtheorem{example}[definition]{Example}

\newtheorem{theorem}[definition]{Theorem}

\newtheorem*{theorem*}{Theorem}
\newtheorem{lemma}[definition]{Lemma}
\newtheorem*{lemma*}{Lemma}
\newtheorem{remark}[definition]{Remark}

\newtheorem{corollary}[definition]{Corollary}
\newtheorem*{corollary*}{Corollary}

\newtheorem{proposition}[definition]{Proposition}

\def\XXint#1#2#3{{\setbox0=\hbox{$#1{#2#3}{\int}$}
		\vcenter{\hbox{$#2#3$}}\kern-.5\wd0}}

\makeatletter
\@addtoreset{definition}{section}
\@addtoreset{equation}{section}
\makeatother

\newcommand{\I}{\mathcal{I}}

\DeclareMathOperator{\Res}{Res}

\allowdisplaybreaks[4]

\title{Universal Correlators on Exponentially Ramified Spectral Curves}

\author{Mohamad Alameddine}

\address{Section of Mathematics, University of Geneva, Rue du Conseil-G\'{e}n\'{e}ral 7-9,
	1205 Geneva, Switzerland}

\email{mohamad.alameddine@unige.ch}

\author{Alexander Hock}

\address{Section of Mathematics, University of Geneva, Rue du Conseil-G\'{e}n\'{e}ral 7-9,
	1205 Geneva, Switzerland} 
\email{alexander.hock@unige.ch}

\begin{document}
\maketitle

\begin{abstract}
We investigate generalized topological recursion on compact spectral curves admitting exponential, and more generally essential, singularities as ramification points. Exploiting the global formulation of generalized topological recursion, we establish a contour deformation of the recursive residue formula that replaces contributions from these essential singularities by residues at meromorphic points. This provides a natural recursive framework for exponentially ramified spectral curves while remaining entirely within the generalized topological recursion formalism. Our formalism can also be viewed as a limiting case of the Bouchard–Eynard higher-order topological recursion, obtained when the order of a ramification point tends to infinity in a convergent manner, as occurs, for example, for exponential singularities while $dy$ remains regular and non-vanishing. We further illustrate the resulting formalism through several examples, including transcendental functions and the $x$-$y$ dual of the Mirzakhani curve.
\end{abstract}

\section{Introduction}

\subsection{History and motivation}

Topological recursion (TR) is a universal formalism that defines a family of multi-differentials $\omega_{g,n}$ on a given spectral curve. While its original formulation was motivated by loop equations coming from matrix models \cite{Eynard:2007kz}, it admits applications in various fields of mathematics including Gromov-Witten theory \cite{Bouchard:2007ys,Eynard:2012nj,Norbury2014,Fang2019}, intersection theory on moduli spaces of curves \cite{Eynard:2007kz,Eynard:2011ga,DuninBarkowski2014}, Hurwitz theory \cite{Bouchard2008,Eynard2011Hurwitz,Dunin-Barkowski:2019iog}, and integrable systems \cite{Eynard:2021sxg,Iwaki:2019zeq,KPandTaufunctions,Marchal:2019bia} just to mention a few.

The spectral curve for the initial data of TR is taken to be a Riemann surface $\Sigma$ along with two (usually) meromorphic functions $x$ and $y$ and a symmetric bi-differential denoted $B$. The original TR was defined near the ramified points of $x$ and required them to be simple and regular points for the function $y$ and non-vanishing for $dy$. The quest for lifting these conditions has generated a gradually growing amount of research leading to different generalizations of TR \cite{Bouchard:2012yg,Do:2014ncn,Chekhov2019,Alexandrov:2023tgl,Alexandrov:2024tjo}. The main protagonist of this article is the \textit{Degenerate and irregular TR} of \cite{Alexandrov:2024tjo} (referred to as generalized TR or in short Gen-TR).


As mentioned above, the original setting of topological recursion consists of simple ramification points. The recursive definition of $\omega_{g,n}$ depends crucially on the fact that the ramification points are simple. A definition including higher-order ramification points, say of order $r_a$ for a ramification point $a$, where $r_a=2$ corresponds to a simple ramification point, was studied in Bouchard–Eynard topological recursion (BE-TR) \cite{Bouchard:2012yg}. The main property was that their definition commutes with taking limits in which colliding ramification points generate a higher-order ramification point.
The construction used in \textit{op. cit.} involves a kind of globalization of the recursion which, however, still depends strongly on the order $r_a$ of the ramification point through the local involutions associated with $r_a$. Another important assumption is that $y$ is regular at the ramification point and that $dy$ is non-vanishing there. In this set-up, the multi-differentials $\omega_{g,n}$ generated by BE-TR admit poles at a ramification point $a$ of order at most $r_a(2g-2+n)+2g$.
From this simple observation, one immediately concludes that, in a limit $x_N\to x$ in which a ramification point acquires infinite order, the multi-differentials can no longer remain meromorphic. In other words, if one requires $y$ to be regular at an infinite-order ramification point and $dy$ to be non-vanishing, then all $\omega_{g,n}$ will develop poles of infinite order, or rather essential singularities. if the limit is convergent.

A first attempt to study ramification points of infinite order was made in \cite{Bouchard:2023yau} under the name of compact trans-algebraic spectral curves. The idea is to consider a sequence of compact meromorphic spectral curves
\[
S_N=\left(\Sigma,x_N,y_N,B\right),
\]
such that $x_N\to x$ and $y_N\to y$ as $N\to\infty$, where $x$ and $y$ are trans-algebraic functions on $\Sigma$ (potentially admitting exponential singularities) with $xy$ meromorphic. The main construction is then captured by the commutation of topological recursion with the infinite limit
\begin{align}
\lim_{N\to\infty}\bigl(\omega_{g,n}^{N}(S_N)\bigr)
=
\omega_{g,n}\!\left(\lim_{N\to\infty}S_N\right),
\end{align}
where $\omega_{g,n}^{N}$ are the TR differentials associated with the initial data $S_N$.
The reason for imposing the condition of $xy$ being meromorphic comes from the enumerative interpretation of these curves in terms of Atlantes Hurwitz numbers. This important structure, however, breaks the crucial assumption that $y$ be regular at the infinite-order ramification point. Loosely speaking, the function $y_N$ admits a pole of the same order as the ramification point of $x_N$. This secretly places us in the setting of \textit{irregular topological recursion} in the sense of Chekhov–Do–Norbury \cite{Do:2014ncn,Chekhov2019} (or even more general), where it was observed that the order of the pole of $\omega_{g,n}$ is of order $2g$, thus independent of $n$ and, in particular, independent of the order of the ramification point itself.
This means that the authors of \cite{Bouchard:2023yau} remain in a setting in which the multi-differentials $\omega_{g,n}$ stay meromorphic in the limit due to the singular behavior of $y$ at the infinite-order ramification point. However, this is not the original setting \cite{Bouchard:2012yg} in which one genuinely changes the order of a ramification point while keeping $y$.

We present here another approach to the problem of exponential or essential singularities. We make no assumption on the product $xy$ and consider the case where one of the two functions admits an exponential or essential singularity, although the case in which both functions admit essential singularities could be treated in a similar way. In the case of an essential singularity in $x$, a typical example is obtained when
\[
x=\lim_{N\to\infty} x_N.
\]
The opposite case, namely an essential singularity in $y$, can be treated equally well through the $x$-$y$ swap formalism \cite{Hock:2022wer,Alexandrov:2022ydc}, which we also consider.
The main result of this work relies on the global nature of the integrand in Gen-TR, which gives rise to globally defined differentials assuming the spectral curve to be compact. Importantly, the recursive definition in Gen-TR is independent of the order of the ramification points. It reduces to BE-TR, whose definition depends on the order of the ramification points, under the assumptions that $y$ is regular and that $dy$ is non-vanishing at those points.

In fact, let us examine the first non-trivial step in the recursive definition of $\omega_{g,n}$, namely $\omega_{0,3}$. It was shown in \cite{Eynard:2007kz} that the local definition of the original topological recursion is equivalent to
\begin{align} \label{Omega03}
    \omega_{0,3}(z_1,z_2,z_3)
    =
    \sum_{p \in \mathcal{P}}
    \Res_{q \to p}
    \frac{B(z_1,q)\,B(z_2,q)\,B(z_3,q)}
         {dx(q)\,dy(q)}.
\end{align}
This formula already appeared, in a related form, in the work of Krichever \cite{Krichever:1992qe}, where it is connected to third derivatives of the genus-zero free energy, in the work of Dubrovin related to Frobenius manifolds \cite{Dubrovin:1994hc}, and also in the work of Bertola \cite{Bertola:2005he}. The expression is symmetric in its three variables and, more importantly, admits a global formulation that is independent of the order of the ramification points, despite of being equivalent to the local definition in topological recursion. Furthermore, the local definition of higher-order BE-TR, which explicitly depends on the orders $r_a$ of the ramification points $a\in\mathcal{P}$, can also be rewritten in the form \eqref{Omega03}.

This already hints at a global recursive definition of topological recursion in which the multi-differentials are defined independently of the order of the ramification points. The set $\mathcal{P}$, introduced later, denotes the collection of points (typically the ramification points of $x$) on which the recursion is established. This motivates a new perspective on topological recursion, in which the assumptions on the initial data $dx$ and $dy$ are relaxed (while keeping them as global meromorphic differentials), leading to an alternative recursive formalism. From another point of view, one may say that Gen-TR \cite{Alexandrov:2024tjo} generalizes the formula \eqref{Omega03} to arbitrary $(g,n)$.

 \subsection{Main ideas and organization}

 Our approach to the problem is based on the formulation of Gen-TR, for this, we will treat a (compact) spectral curve with an essential singularity as a limit of sequences of meromorphic (compact) spectral curves and prove that essential singularities can be covered in this framework through a contour deformation of the globally defined integrand of the recursion. Let us illustrate this idea first on a simple example, in particular, a contour deformation of the residue of the following one form 
\begin{align}
    \underset{ z \to \infty}{\Res} \frac{e^z}{(z-t)^2} d z
\end{align}
Using series expansions of both functions in the numerator and denominator, one could carry out the computation in an explicit way 
\begin{align}
     \underset{ z \to \infty}{\Res} \frac{e^z}{(z-t)^2} d z =& - \underset{ z \to \infty}{\Res} \sum_{m,n \geq 0} \frac{(n+1) t^n}{m !} z^{m-n-2} =  - \sum_{n=0}^\infty \frac{t^n}{n !} = - e^t.
\end{align}
However, one realizes that in order to compute such a residue, one needs the infinite Laurent (series) of the components inside the residue. Another simpler way exists in order to compute such a residue, via a contour deformation. One has 
\begin{align}
    \underset{ z \to \infty}{\Res} \frac{e^z}{(z-t)^2} d z =& - \underset{ z \to t }{\Res} \frac{e^z}{(z-t)^2} dz = - \frac{d}{dz} e^z \bigg\vert_{z=t} = -e^t
\end{align}
One realizes the easy computation when one moves the contour of integration, this example illustrates the idea but does not limit its use. In other words, this idea is not limited to the above example as we will see. Note that one important ingredient allowing for such a deformation is the globally defined integrand, this is a crucial property that is offered by Gen-TR as we will see. An additional example, this time closer to TR and to the discussion above is the explicit universal expression for $\omega_{0,3}$ given in \eqref{Omega03}. Indeed, already on a global curve, the idea of moving the contour can already be used in order to express expression \ref{Omega03} as 
\begin{align}
   \omega_{0,3} (z_1,z_2,z_3) = -\sum_{p \in \{z_1,z_2,z_3, \mathcal{P}^\vee \} }\underset{q \to p}{\Res} \frac{B(z_1,q) B(z_2,q) B(z_3,q)}{d x(q) d y (q)}
\end{align}
assuming that the Bergman kernel $B$ is normalized along some choice of $\mathcal{A}$-cycles. If $x$ has an essential singularity in the set $\mathcal{P}$, this expression is well-defined and can easily be computed by the residues at $\{z_1,z_2,z_3, \mathcal{P}^\vee \}$. The new expression picks up poles at the diagonals and the ramification points of $y$ (which is here the set $\mathcal{P}^\vee$), and this is what we will establish for all the differentials $\omega_{g,n}$. In other words, we will show that the recursion of Gen-TR producing differentials is equivalent to the following recursion  
 \begin{align}
       \omega_{g,n} (z, z_{\llbracket n-1  \rrbracket}) = -  \sum_{q \in \{z,z_i ,\mathcal{P}^\vee \}} \underset{\tilde{z} = q}{\Res}  \int^{\tilde{z} } B(\cdot,z) \,\, \bar{\omega}_{g,n} (\tilde{z}, z_{\llbracket n-1  \rrbracket})
\end{align}
Let us explain briefly the concept behind the sets of points $\mathcal{P}$ and $\mathcal{P}^\vee$. In fact, for the original definition of TR, the set $\mathcal{P}$ consists of the ramification points of $x$ with some constraints on the order of these points following different reformulations. However for Gen-TR, it can more generally be taken as a subset of a set of special points (essentially the set of ramification points of $x$ and $y$ plus some further points, see Def.~\ref{def:specialpoint}), the complement of this subset is denoted $\mathcal{P}^\vee$ and consists of the points on which the $x$-$y$ dual recursion runs. More precisely, the definition of Gen-TR itself includes a set of special points $\mathbb{S}= \{ \mathcal{P} \sqcup \mathcal{P}^\vee \}$ and a choice of the partition of these points into the complementary subsets $\mathcal{P}$ and $\mathcal{P}^\vee$. The recursion itself runs on the set of points $\mathcal{P}$ while the dual one runs on the complement $\mathcal{P}^\vee$. Therefore, the definition of Gen-TR is adapted for the $x$-$y$ duality linking two different families of differentials defined on dual initial data. We will recall this in Sec.~\ref{Sec2} since we will use this in the examples.   

Effectively, one needs a new point of view on TR in order to have globally defined differentials, and in order to be able to treat curves coming from different problems of Hurwitz theory. As mentioned above, we will recall in this article the new point of view on TR introduced in \cite{Alexandrov:2024tjo}, this is the main subject of Sec.~\ref{Sec2}. We will start this section by recalling the definition of the original TR of \cite{Eynard:2007kz} and review how one could lift the assumptions on the initial data leading to Gen-TR. We will then discuss the mechanism of this new recursion which runs on a set of points called key points $\mathcal{P}$ subset of a certain set of special points. The properties of the differentials allow us later to establish a contour deformation of the residue formula. While this deformation will add residual points, the main advantage behind it is that choosing the set of key points $\mathcal{P}$ allows us to avoid certain residues. We will discuss this in detail and provide some preliminary examples to illustrate our point. Perhaps, one way to put our main statement is that Gen-TR naturally accommodates exponential (essential) singularities through contour deformation while remaining compatible with the $x$-$y$ duality. 

This approach is primarily motivated by the question of curves admitting essential singularities, one famous example is the $x$-$y$ dual side of the Mirzakhani curve. We will show how one can consider these examples using our approach when we discuss essential singularities at the end of Sec.~\ref{Sec2}. This will offer an additional computational tool when considering explicit expressions for the differentials. We will then move to Sec.~\ref{Sec3} where we consider several such examples in increasing difficulty. The main aim of this section is to display the various techniques (from logarithmic TR, $x$-$y$ duality and contour deformation) used in order to handle spectral curves with essential singular behavior. We first of all consider a spectral curve whose different interpretations (parameterization) land in two different frameworks, one is the so-called logarithmic TR (an extension of TR to handle logarithmic singularities of spectral curves), and the other is a curve with an essential singularity which we treat using the $x$-$y$ duality. We will show that the choice of Bergman kernel (or better the parameterization) will affect the differentials constructed. We will then consider a transcendental curve related to the Riemann $\zeta$-function, and finally show that our methods extend to include curves such as the Mirzakhani spectral curve related to Weil-Petersson volumes.

The final Sec.~\ref{Sec4} is devoted to open questions and directions in which essential singularities may naturally arise, we will recall in particular the importance of these methods and frameworks when considering the quantization of these curves. In fact, the original TR and its extensions admit a quantization scheme relating it for instance to integrable systems as Lax representations of differential type.

We leave the treatment of specific curves from different branches of Hurwitz theory or Gromov-Witten invariants as an open question that we hope to consider in future works.

\subsection{Notations and reminders}
We will use the following notations
\begin{itemize}
    \item $\llbracket n  \rrbracket = \{ 1,\dots,n\}$ for any $n \geq 1$ and we set by convention $\llbracket 0 \rrbracket = \emptyset$.
    \item The notation $z_I$ is used for $\left( z_i \right)_{i \in I}$ for any finite and non-empty set or multiset $I$.
    \item The operator $[u^d]$ multiplied on the left of a series expansion is the operator that extracts the $d^{\text{th}}$ coefficient in the expansion in the variable $u$, that is $[u^d] \left( \sum_{k=-\infty}^{+\infty} f_k u^k \right) := f_d$.
    \item The function $\mathcal{S}(u)$ is defined as the function $\mathcal{S}(u):=\frac{2\sinh(\frac{u}{2}) }{u}= \frac{e^{\frac{u}{2}} - e^{-\frac{u}{2}} }{u}$.
    \item The notation $\big\vert_{\tilde{z} \to z}$ denoted the substitution operator $\tilde{z} \to z$, that is $\big\vert_{\tilde{z} \to z} f(\tilde{z}) = f(z) $
\end{itemize}
Furthermore, since we are dealing with multi-differentials on the cartesian product of a compact Riemann surface, we need to introduce \textit{Riemann's bilinear identity}: consider a homology basis of the genus $g$ Riemann surface $(\mathcal{A}_i, \mathcal{B}_i)_{i=1}^g$, if $\omega_1$ and $\omega_2$ are two meromorphic forms on the curve, consider the primitive $\Phi_1(p)$ of $\omega_1$ given by 
\begin{align}
    \Phi_1(p) := \int _{p_0}^p \omega_1.
\end{align}
The primitive depends on the choice of the base point $\{p_0\}$. One has the identity
\begin{align}
    \underset{p \to \text{all poles}}{\Res} \,  \Phi_1(p) \omega_2 (p) = \frac{1}{2 \pi i} \sum_{i=1}^g \oint_{\mathcal{A}_i}  \omega_1 \oint_{\mathcal{B}_i}  \omega_2 - \oint_{\mathcal{B}_i}  \omega_1 \oint_{\mathcal{A}_i}  \omega_2.  
 \end{align}

\subsection*{Acknowledgment}
M. A. was partially funded by the SNSF under the grant agreement PZ00P2 223297 and by the DFG -- project ID 551478549.  A.H. was funded by the Swiss National Science Foundation (SNSF) through the
Ambizione project “TRuality: Topological Recursion, Duality and Applications”
under the grant agreement PZ00P2 223297 and partially by the DFG -- project ID 551478549.

\section{The new view on topological recursion} \label{Sec2}
We will review the recent understanding of Alexandrov et al. \cite{Alexandrov:2024tjo} on TR, by providing first the setup of the original TR \cite{Eynard:2007kz} and then arguing the need for a generalized setup ultimately leading to Gen-TR. 

\subsection{Original Definition of TR (EO-TR)}
We start by recalling the original TR proposed by \cite{Eynard:2007kz}, for this, fix $(\Sigma,x,y,B)$ to be the spectral curve data. In one of the most basic settings, there are no assumptions on $\Sigma$ and it can be, for instance, a collection of small discs, both $x$ and $y$ are assumed being meromorphic functions on $\Sigma$\footnote{This is the original setup defined in \cite{Eynard:2007kz}, which however also works if $dx$ and $dy$ are meromorphic, meaning $x,y$ are allowed to admit logarithmic singular points. This is exactly the case in topological string theory with mirror curve as the TR spectral curve \cite{Bouchard:2007ys} computing Gromov-Witten invariants of toric Calabi-Yau threefolds}, all zeroes of $dx$ are simple, and $dy$ is holomorphic and non-vanishing at zeroes of $dx$, and $B$ is a meromorphic bi-differential on $\Sigma^2$ with the only pole along the diagonal with bi-residue 1. In this setup, TR defines the $n$-differentials (multi-differentials) $\omega_{g,n}$ on $\Sigma^n$ by 
\begin{align}
    \omega_{0,1} = y d x, \qquad \omega_{0,2} = B
\end{align}
and for $2g -2 +n \geq 0$ recursively by a two step procedure. Now, let $z$ and $z_j$ be the same local coordinate around a zero point of $dx$. The first step consists in defining the germs of differentials
\begin{align}
    \bar{\omega}_{g,n+1 \vert i} (z, z_{\llbracket n  \rrbracket}) = & \frac{1}{\left( y(z ) - y(\sigma_i(z) ) dx(z) \right)} \bigg(  \omega_{g-1,n+2} (z, \sigma_i(z) , z_{\llbracket n  \rrbracket}) \nonumber  \\\label{EOTRtdom}
     &   \sum_{\substack{g_1 + g_2 = g, \,\,\,  I_1 \bigsqcup I_2 = \llbracket n  \rrbracket  \\ (g_i, \vert I_i \vert  ) \neq (0,0) }}  \omega_{g_1, \vert I_1 \vert +1  } (z, z_{I_1}) \omega_{g_2, \vert I_2 \vert +1 } (\sigma_i(z), z_{I_2})      \bigg)
\end{align}
in the vicinity of the points $q_i$ with $i=1,...,N$ and for the second step consists in taking residues
\begin{align}\label{originalTR}
    \omega_{g,n+1} (z, z_{\llbracket n  \rrbracket}) = \sum_{i=1}^N \underset{ \tilde{z } = q_i }{\Res}  \int^{\tilde{z}} B (. , z) \, \bar{ \omega }_{g,n+1 \vert i} ( \tilde{z}, z_{\llbracket n  \rrbracket}) .
\end{align}
The integration constants in the last formula can be chosen arbitrarily because $\bar{\omega}_{g,n+1 \vert i}$ has trivial residues at the $q_i$. In  the above two-step procedure, the set $\{ q_1,\dots q_N\}$ is the set of zeroes of $d x$ with local coordinate $z$ and $z_j$, and the function $\sigma_i$ is the deck transformation of $x$ near the point $q_i$. We want to emphasize that $\int^{\tilde{z}} B (. , z)$ together with the residue $\underset{ \tilde{z } = q_i }{\Res}$ taken is a projection operator on $\bar{ \omega }_{g,n+1 \vert i}$, which projects on its principle part in the Laurent series of the variable $\tilde{z}$ at the point $q_i$. From this perspective, the differentials $  \bar{\omega}_{g,n+1 \vert i}$ are seen as a local approximation of the TR differentials $\omega_{g,n+1}$ at those points $q_i$. More precisely, we have
\begin{align}
    \omega_{g,n+1} (z, z_{\llbracket n  \rrbracket})=   \bar{\omega}_{g,n+1|i} (z, z_{\llbracket n  \rrbracket}) + \, \text{holomorphic part}, \quad z \to q_i, \,\,\,\forall i =1, \dots,N.
\end{align}
Respectively, the integral transformation then recovers $\omega_{g,n+1}$ as a $1$-form in the first coordinate on $\Sigma$, since $B$ is defined globally on $\Sigma$. All $\omega_{g,n}$ with $2g-2+n>0$ are thus globally meromorphic, having no poles other than at the zeros of $dx$. Moreover, the principal parts of their poles at the points $q_i$ coincide with those of $\bar{\omega}_{g,n+1|i}$. These differentials inherit uniquely the properties of $B$. This means that, for instance, on a compact Riemann surface $\Sigma$ equipped with a choice of $\mathcal{A}$- and $\mathcal{B}$-cycles, the Bergman kernel $B$ can be uniquely fixed by imposing the normalization condition along the $\mathcal{A}$-cycles. In this case, all $\omega_{g,n}$ inherit this property from $B$.

\subsection{Definition of Gen-TR}
Some different generalizations of TR exist, formulated for different reasons, applications, handling different types of singularities or ramification data etc. Clearly, one has in the original TR the summation over the simple zeroes of the form $dx$, and one natural generalization would be to lift this recursion to a higher recursion where the order of the zeroes of $dx$ is not anymore simple. This is the main motivation that leads to a form of TR that makes sense in a much more general setting: it will be sufficient just to assume that $dx$ and $dy$ are arbitrary meromorphic differentials on $\Sigma$. This, of course, provides a different algebraic expression for the $\bar{\omega}_{g,n+1|i}$ since the order of the ramification point is used in its definition, for instance in the local deck transformation $\sigma_i$. 

In \cite{Bouchard:2012yg}, a local higher order definition for $\bar{\omega}_{g,n+1|i}$ was provided depending on the order of the ramification point. We, however, want to use another quite new definition called \textit{Gen-TR}, which is defined by a global integrand $\bar{\omega}_{g,n+1}$ independent of the order of the ramification points, which locally reduces to the $\bar{\omega}_{g,n+1|i}$ of EO-TR in some setup if the ramification points of $x$ are simple, or to BE-TR \cite{Bouchard:2012yg} if the ramification points of $x$ have higher order plus some additional assumptions.  Furthermore, Gen-TR actually allows an additional choice whether or not certain points are chosen to be poles of the $\omega_{g,n}$ or not. This means that the following definition is valid in a much more general setting and does only give the same $\omega_{g,n}$ as EO-TR if we restrict to its set-up. Thus, we have much more freedom. Precise statements will follow after the definition.

As for any recursion, let us start by the initial data needed. 
\begin{definition}\label{def:specialpoint}
    Let $dx$ and $dy$ be two meromorphic differentials on a smooth complex curve $\Sigma$, let $q \in \Sigma$ and $z$ be a local coordinate at this point. Then the local expansion of $dx$ and $dy$ at this point is given by 
    \begin{align}
        dx= a z^{r-1} \left( 1 + O(z) \right) dz, \qquad  dy= b z^{s-1} \left( 1 + O(z) \right) dz, \qquad a,b \neq 0 , \,\,\,\,\, r,s \in \mathbb{Z}.
    \end{align}
    The point $q$ is called non-special if $\begin{cases} 
         r =s=1 \,\,\,\, or \\
          r+ s \leq 0
    \end{cases}$  and special otherwise.
\end{definition}
For example, zeros of $dx$ at which $dy$ is regular and non-vanishing are special points; similarly, zeros of $dy$ at which $dx$ is regular and non-vanishing are special points.

The initial data of Gen-TR is defined as follows
\begin{definition}
    The initial data of Gen-TR is given by the tuple $(\Sigma, dx, dy, B, \mathcal{P} )$ where 
    \begin{itemize}
        \item $\Sigma$ is a smooth complex curve. 
        \item $B$ is a symmetric bi-differential on $\Sigma^2$ with a second order pole on the diagonal with bi-residue $1$ and no other poles.
        \item $dx$ and $dy$ are two arbitrary non-zero meromorphic differentials on $\Sigma$
        \item $\mathcal{P} = \{ q_1,\dots q_N \}$ is an arbitrarily chosen finite subset in the set of special points. The points of $\mathcal{P}$ are called key points. 
    \end{itemize}
\end{definition}
    
One realizes immediately that there are no restrictions imposed on the orders of poles and zeroes of $dx$ and $dy$. Before giving the definition of Gen-TR, it is useful to define for a given collection of meromorphic differentials $  \omega_{g,n}$ such that $  \omega_{0,1} =  yd x$ the formal series
\begin{align} \label{collection}
    \omega_n = \sum_{ \substack{g \geq 0 \\ (g,n) \neq (0,1) }} \hbar^{2g - 2 +n}   \omega_{g,n}. 
\end{align}
Now, assuming that the differentials $\omega_n$ are globally defined meromorphic $n$-differentials (which holds by induction), we consider a new collection of differentials $\mathcal{W}_n(z,u ; z_1,\dots z_n) = \sum_{g=0}^\infty  \hbar^{2g - 1 +n}  \mathcal{W}_{g,n}$ depending on an additional formal parameter $u$ for $n \geq 0$ by the following relations 
\begin{align}
    \mathcal{T} _n (z,u; z_{\llbracket n  \rrbracket}) = &  \sum_{k=1}^\infty  \frac{1}{k !}  \prod_{i=1}^k \left( \bigg\vert_{\tilde{z}_i \to z} u \hbar \mathcal{S}(u \hbar \partial_{\tilde{x}_i}) \frac{1}{d \tilde{x}_i}   \right) \left(  \omega_{k+n} (\tilde{z}_{\llbracket k  \rrbracket}, z_{\llbracket n  \rrbracket})   -  \delta_{n,0} \delta_{k,2} \frac{ d \tilde{x}_1  d \tilde{x}_2   }{\left( \tilde{x}_1 - \tilde{x}_2 \right)^2} \right) \nonumber \\\label{TDefinition}
    &  +\delta_{n,0} u \left(\mathcal{S} (u \hbar \partial_x) -1 \right) y,        \\\label{WDefinition}
    \mathcal{W} _n (z,u; z_{\llbracket n  \rrbracket}) =  & \frac{dx}{u \hbar} e^{\mathcal{T}_0 (z,u)} \sum_{\substack{\llbracket n  \rrbracket = \bigsqcup_\alpha J_\alpha  \\ J_\alpha \neq 0    }  } \prod_\alpha   \mathcal{T}_{\vert J_\alpha \vert} (z,u;z_{ J_\alpha })       
\end{align}
where $\tilde{x}_i = x (\tilde{z}_i), x=x(z)$ and $y=y(z)$. From the definition itself, one can prove the following set of properties satisfied by the differentials (see \cite{Alexandrov:2024tjo} for details)
\begin{itemize}
    \item $\mathcal{W}_{g,n}$ is a polynomial combination of the differentials $dx, dy$ and $\omega_{g',n'}$ and their derivatives (in particular, the function $y$ enters with derivatives of positive order only, and the explicit dependence of $\tilde{x}_1,\tilde{x}_2 $ in the term with $ \delta_{n,0} \delta_{k,2}$ disappears as well). As a consequence, it is a global meromorphic $(n + 1)$-differential on $\Sigma^{n+1}$. This is an important property hat we will use for establishing our result later. 
    \item We have $\mathcal{W}_{0,0} = [\hbar^{-1} ] \mathcal{W}_{0} = \frac{d x}{u}$, and for $(g,n) \neq (0,0)$, the dependence of $\mathcal{W}_{g,n}$ on $u$ is polynomial.
    \item For each $(g,n) \neq (0,0)$, one has $ [u^0] \mathcal{W}_{g,n} = \omega_{g,n+1}$.
    \item It is useful to present $\mathcal{W}_{g,n}$ for $(g,n) \neq (0,0)$ in the following way 
    \begin{align}
        \mathcal{W}_{g,n} (z,u; z_{\llbracket n  \rrbracket}) =  \omega_{g,n+1}(z, z_{\llbracket n  \rrbracket}) + \bar{\mathcal{W}}_{g,n} (z,u; z_{\llbracket n  \rrbracket})
    \end{align}
    where $\bar{\mathcal{W}}_{g,n}$ is the contribution of the terms with positive exponents of $u$, more precisely $\bar{\mathcal{W}}_{g,n}= \sum_{k \geq 1} [u^k] \mathcal{W}_{g,n}$. Then, the expression of $\bar{\mathcal{W}}_{g,n}$ involves the differentials $\omega_{g',n'}$ for $2g' -1 +n' < 2g -1 +n$ only.
\end{itemize}
This ultimately leads us to the definition of Gen-TR given as follows
\begin{definition} \label{Gen-TRDef}
    The differentials of Gen-TR $\omega_{g,n}$ with $g \geq 0, n \geq 1$ for the set of initial conditions $(\Sigma, dx, dy, B, \mathcal{P} )$ are defined in the unstable cases by
    \begin{align}
        \omega_{0,1} (z) = y(z ) d x(z), \qquad \omega_{0,2} (z_1,z_2) = B(z_1,z_2)
    \end{align}
    and for $2g - 2 + n > 0$ 
    \begin{align}
        \omega_{g,n} (z, z_{\llbracket n-1  \rrbracket}) = &\sum_{q \in \mathcal{P} } \underset{\tilde{z}= q}{\Res} \,\,  \int^{\tilde{z} } B(\cdot,z) \,\, \bar{\omega}_{g,n} (\tilde{z}, z_{\llbracket n-1  \rrbracket}), \qquad \text{where, one has } \nonumber \\ 
        \bar{\omega}_{g,n} (z, z_{\llbracket n-1  \rrbracket}) = & - \sum_{r \geq 1} \left( - d \frac{1}{d y} \right)^r [u^r] \, \bar{\mathcal{W}}_{g,n-1} (z, u; z_{\llbracket n-1  \rrbracket}). 
    \end{align}
\end{definition}
Some comments are in order on this loaded definition (see \cite{Alexandrov:2024tjo}). First, like the differentials of the original TR the new differentials are symmetric in all $n$ arguments for any choice of key points $\mathcal{P}$. Second, if the set of key points is taken to be the set of zeroes of $dx$ and the initial spectral curve data satisfies the original TR setup (all ramification points of $x$ are simple), the differential generated by Def.~\ref{Gen-TRDef} are identical to the ones defined by the EO-TR of equation \eqref{originalTR}. The reason is literarily 
\begin{align}\label{baromgi}
    \bar{\omega}_{g,n} (z, z_{\llbracket n-1  \rrbracket})=\bar{\omega}_{g,n|i} (z, z_{\llbracket n-1  \rrbracket}) + \, \text{holomorphic part}, \quad z \to q_i, \,\,\,\forall i =1, \dots,N
\end{align}
where $\bar{\omega}_{g,n|i} (z, z_{\llbracket n-1  \rrbracket})$ is defined by \eqref{EOTRtdom} and $q_i$ are the ramification points of $x$. The property \eqref{baromgi} is however highly nontrivial.  Gen-TR has the benefit of having a universal global integrand $\bar{\omega}_{g,n}$ independent of the ramification point $q_i$. If the ramification points have higher order and all ramification points of $x$ are chosen in the set $\mathcal{P}$, then the $\omega_{g,n}$ of Gen-TR coincide with the $\omega_{g,n}$ of \cite{Bouchard:2012yg} if $s$ in Def.~\ref{def:specialpoint} is $\pm 1\, \text{mod} \,r$, see also \cite{Borot:2023wik}. Thus, the integrand $ \bar{\omega}_{g,n}$ is also independent of the order of the ramification points. Note also that $\bar{\mathcal{W}}^{(g)}_{n-1} (z, u; z_{\llbracket n-1  \rrbracket})$ has just poles in $z$ at the special points $(=\mathcal{P}\sqcup\mathcal{P}^\vee)$ (where $\mathcal{P}^\vee$ is the complementary set to the set of key points composed of $\vee$-key points in the set of special points) and at the diagonal $z=z_i$. Furthermore, the differentials are determined uniquely in the case where $\Sigma$ is compact and $B$ is the Bergman kernel normalized along the $\mathcal{A}$-cycles for some choice of $\mathcal{A}$- and $\mathcal{B}$-cycles. 

Since the recursive formula of Gen-TR is entirely formulated in terms of globally defined forms, it is possible to provide a further equivalent residue computation on a compactly connected Riemann surface $\Sigma$:
\begin{lemma} \label{contourdef}
    Let $(\Sigma,dx,dy,B,\mathcal{P})$ be the initial data of Gen-TR with $\Sigma$ compact, connected, and $dx,dy$ globally defined. Furthermore, let $B$ be the normalized Bergman kernel for some canonical choice of $\mathcal{A}$- and $\mathcal{B}$-cycles. Then, the recursion of Def.~\ref{Gen-TRDef} is equivalent to the following recursion 
    \begin{align}\label{Lemmaeq1}
       \omega_{g,n} (z, z_{\llbracket n-1  \rrbracket}) = -  \sum_{q \in \{z,z_{\llbracket n-1  \rrbracket} ,\mathcal{P}^\vee \}} \underset{\tilde{z} = q}{\Res}  \int^{\tilde{z} } B(\cdot,z) \,\, \bar{\omega}_{g,n} (\tilde{z}, z_{\llbracket n-1  \rrbracket})
    \end{align}
     where $\bar{\omega}^{(g)}_n (z, z_{\llbracket n-1  \rrbracket})$ is again given by 
    \begin{align}\label{Lemmaeq2}
        \bar{\omega}_{g,n} (z, z_{\llbracket n-1  \rrbracket}) = & - \sum_{r \geq 1} \left( - d \frac{1}{d y} \right)^r [u^r] \, \bar{\mathcal{W}}_{g,n-1} (z, u; z_{\llbracket n-1  \rrbracket}). 
    \end{align}
    \begin{proof}
    The starting point of the proof is an interpretation of the recursion formula, in fact, the recursion works as a residue at the set of key points $q \in \mathcal{P}$ of a projection operator $\Res \int B$ applicable to a Laurent series and extracting its polar part around each of the points. The two pieces are global on the cartesian product of the Riemann surface $\Sigma$. The projection operator admits poles at the point $q=z$, while $\bar{\mathcal{W}}^{(g)}_{n-1}(q, u; z_{\llbracket n-1  \rrbracket})$ (or better $\bar{\omega}_{g,n}(q, z_{\llbracket n-1  \rrbracket})$) admits poles at the points $q \in \{  \mathcal{P}, \mathcal{P}^\vee, z_{\llbracket n-1  \rrbracket}  \}$ where $z_i\in z_{\llbracket n-1  \rrbracket}$ are poles at the diagonals. Thus a contour deformation of Def.~\ref{Gen-TRDef} would provide the alternative expression for the recursion
\begin{align}\nonumber
  &- \frac{1}{2 \pi i} \sum_{i=1}^g \bigg(\oint_{\mathcal{B}_i}   B(\bullet,z) \,\,  \oint_{\mathcal{A}_i}  \bar{\omega}^{(g)}_n (\bullet, z_{\llbracket n-1  \rrbracket})-\oint_{\mathcal{B}_i}    \bar{\omega}^{(g)}_n (\bullet, z_{\llbracket n-1  \rrbracket})\,\,  \oint_{\mathcal{A}_i} B(\bullet,z)  \bigg)\\\nonumber
  &-  \sum_{q \in \{z,z_{\llbracket n-1  \rrbracket}, \mathcal{P}^\vee \}} \underset{\tilde{z} = q}{\Res}  \int^{\tilde{z} } B(\cdot,z) \,\, \bar{\omega}^{(g)}_n (\tilde{z}, z_{\llbracket n-1  \rrbracket}),
\end{align}
where we have used Riemann bilinear identity to express the rest of the poles in the residues. 
Note that the second term in the first line vanishes due to the normalization of $B$ around all $\mathcal{A}_i$-cycles. Furthermore, the first term vanishes since $\bar{\omega}^{(g)}_n$ is constructed recursively by global meromorphic functions and more importantly it is an exact 1-form by definition because it starts in \eqref{Lemmaeq2} with $r\geq1$. Therefore, because the integration of any exact 1-form vanishes over any closed circle on a compact Riemann surface, one gets the desired result.
\end{proof}
\end{lemma}

The lemma can be very useful in situations where the set of key-points $\mathcal{P}$ is complicated or large in comparison to the set of $\vee$-key points $\mathcal{P}^\vee$. Note that the residue around $z$ in \eqref{Lemmaeq1} is easily performed from which the integrand itself is obtained. Thus, we can also write:
\begin{align*}
    \omega_{g,n} (z, z_{\llbracket n-1  \rrbracket}) =\bar{\omega}_{g,n} (z, z_{\llbracket n-1  \rrbracket}) -  \sum_{q \in \{z_{\llbracket n-1  \rrbracket}, \mathcal{P}^\vee \}} \underset{\tilde{z} = q}{\Res}  \int^{\tilde{z} } B(\cdot,z) \,\, \bar{\omega}_{g,n} (\tilde{z}, z_{\llbracket n-1  \rrbracket}).
\end{align*}
In this way, we are actually able to define $\omega_{g,n}$ recursively which have poles in $\mathcal{P}$ without looking at its behavior around those points. The poles of $\omega_{g,n}$ at $\mathcal{P}$ will coincide with those of $\bar{\omega}_{g,n} (z, z_{\llbracket n-1  \rrbracket})$, and the last term literally annihilates all additional poles of $\bar{\omega}_{g,n} (z, z_{\llbracket n-1  \rrbracket})$ not in $\mathcal{P}$.

\subsection{Gen-TR and $x$-$y$ duality} 
An important duality in the context of topological recursion (now known as the $x$-$y$ duality) was understood only recently through a sequence of works \cite{Bychkov:2020ujd,Bychkov:2020yzy,Borot:2021thu,Hock:2022wer,Hock:2022pbw,Alexandrov:2022ydc,Hock:2023dno,Alexandrov:2023tgl}, which motivated the definition of Gen-TR \cite{Alexandrov:2024tjo}. What we will use here is its compatibility with Gen-TR.

Let us define a dual initial data by the tuple $(\Sigma, dy, dx, B, \mathcal{P}^{\vee} )$. The complementary subset in the set of special points is denoted $\mathcal{P}^\vee$ which is also finite and consists of points that are referred to as $\vee$-key points. In the context of the $x$-$y$ duality, the notion of key and $\vee$-key points now comes with a choice; first, one notices that the notion of special points (decomposed into key and $\vee$-key points) is symmetric with respect to the swap relation, and in order to define the input, one needs to choose whether the special point is in the key or $\vee$-key points. If there are no common zeroes and no simple poles of $dx$ and $dy$, there exists a canonical choice depending on the zeroes of $dx$ to be the key points and the zeroes of $dy$ to be the $\vee$-key points. This is exactly the original TR choice. However, if there exists an intersection between the zeroes of $dx$ and $dy$, original TR is not defined, but Gen-TR still makes sense with a preference regarding the choice depending on the case at hand, see for instance \cite{Bouchard:2025rid}. \\

Let us start by recalling the setup of the duality, let us first state that the $x$-$y$ duality can be reformulated completely independently of TR. One could define it for any family of symmetric multi-differentials $(\omega_{g,n})_{(g,n)\neq (0,1)}$ on $\Sigma$. However, for the definition to work, one needs this family to be globally defined, meromorphic\footnote{This extends to the case of an essential singularities by the same contour deformation argument of the present work.}, and admits no poles at the diagonals except for the $(0,2)$ element which is taken to be meromorphic admitting a double order pole on the diagonal such that for any local coordinate $z$ the regularized form $\omega_{0,2} (z_1,z_2) - \frac{dx(z_1) dx(z_2)}{\left(x(z_1)-x(z_2) \right)^2}$ is regular on the diagonal. In the original setting \cite{Alexandrov:2024tjo} it was assumed that the differentials $dx$ and $dy$ are meromorphic, which can be relaxed as we will show. We will recall the duality rather in an explicit way independently from any recursion, and then explain the ingredients inside of it.  
\begin{definition} \label{x-yswapdef}
    For a given family of multi-differentials $\{\omega_{g,n}\}_{(g,n)\neq (0,1)}$ and two meromorphic differentials $dx$ and $dy$, the dual family of multi-differentials denoted by $\{ \omega_{g,n}^\vee\}_{(g,n)\neq (0,1)}$ is defined as 
    \begin{align} \label{transformationdual}
      \omega_{n}^\vee (z_{\llbracket n \rrbracket})= (-1)^n \left( \prod_{i=1}^n \sum_{r=0}^\infty \left( - d_i \frac{1}{d y_i} \right)^r [u_i^r]  \right)    \mathbb{W}_n(z_{\llbracket n \rrbracket}, u_{\llbracket n \rrbracket})\, .
    \end{align}
where we define the extended $n$-differentials $\mathbb{W}_n := \sum_{g=0}^\infty \hbar^{2g-2+n} \mathbb{W}_{g,n}$ where
\begin{align}
\mathbb{W}_n = & \mathbb{W}_n(z_{\llbracket n \rrbracket}, u_{\llbracket n \rrbracket}) := \prod_{j=1}^n \frac{dx_j}{u_j \hbar} e^{u_j (\mathcal{S}( u_j \hbar \partial_{x_j})-1)y_j} \sum_\Gamma \frac{1}{\vert \mathrm{Aut}(\Gamma)  \vert} \nonumber \\
& \prod_{\bullet \in W(\Gamma)} \prod_{i \in I(\bullet)}   u_i \hbar \mathcal{S} (u_i \hbar \partial_{x_i})           \frac{\omega_{|I(\bullet)|} \left( z_{ I(\bullet) } \right)}{\prod_{i\in I(\bullet)} d x_i}
\end{align}
and where the $\omega_n$ is defined in \eqref{collection}, we require it to be regularized when $I=\{i,i\}$. For $i\in I(\bullet)$, the differential operator $ \mathcal{S} (u_i \hbar \partial_{x_i})   $
acts argument-wise on \(\omega_{|I(\bullet)|}(z_{I(\bullet)})\). If a label appears several times in the multiset \(I(\bullet)\), each occurrence is treated as a distinct argument of \(\omega_{|I(\bullet)|}\), and the corresponding differential operators act independently.

\end{definition}
Let us explain in details the main factors constituting this definition of a dual family of differentials. First, the extended differentials $\mathbb{W}_n$ contain a sum over all connected bipartite graphs $\Gamma$ with $n$ labeled vertices, unlabeled vertices (both of valency $\geq1$) and multi-edges. This means the following:
\begin{itemize}
    \item A fixed set of labeled vertices $V(\Gamma) =\{1,2 , \dots,n \} $.
    \item A set of unlabeled vertices denoted $W(\Gamma)$, sometimes referred to as black (or $ \bullet$) vertices. 
    \item The graph is connected meaning that between any two vertices there exists a path (possibly using several edges joining them).
    \item Between a labeled and an unlabeled vertex, several edges may connect the same pair of vertices, giving rise to a multi-edge.
     
    \item To an unlabeled vertex $\bullet\in W(\Gamma)$, we associate the multiset $I(\bullet)$ consisting of the labels from $V(\Gamma)$ to which the vertex $\bullet$ is connected. We denote by $|I(\bullet)|$ the valency of the unlabeled vertex $\bullet\in W(\Gamma)$, with multiplicities of multi-edges taken into account.
\end{itemize} 
Automorphisms of a graph $\Gamma$ arise from the multiplicities of multi-edges and from permutations of unlabeled vertices with identical associated multisets, that is, vertices $\bullet,\bullet'\in W(\Gamma)$ satisfying $I(\bullet)=I(\bullet')$.
Moreover, one may realize that the sum over these graphs is infinite, however, the expansion of the exponential in front and the recollection of $\omega_n$ into power series of $\hbar$ ensures that only a finitely many graphs contribute for each $\omega_{g,n}^\vee$.

The multi-differentials $\mathbb{W}_n$ are symmetric in $n$ pairs of variables $(z_i,u_i)$ where $z_i$ is a point of the $i$-th factor in $\Sigma^n$ and $u_i$ is a formal parameter. Note also that apart from the first term $\mathbb{W}_{0,1}$, the expansion provides polynomials in the formal variables $u_1,\dots,u_n$ whose coefficients are meromorphic multi-differentials on $\Sigma^n$. 

This formula has different interpretations and one useful way of manifesting it is the correspondence between moments and cumulants in the higher order free probability theory \cite{Borot:2021thu}. 

\begin{remark}
   This formula appeared in several instances and in various contexts, we are using the combinatorial interpretation of \cite{Hock:2022pbw}, note that this is equivalent to the form of the formula proposed in \cite{Alexandrov:2022ydc}. In fact, additional combinatorial structure and insight can be extracted, as shown in \cite{Hock:2022pbw}.
\end{remark}

The theorem is finally that the swap is an \textit{involution} or in other words a duality \cite{Alexandrov:2022ydc}, one could start from the dual side in an analogous matter and define the dual differentials
\begin{align}
\mathbb{W}_n^\vee = & \mathbb{W}^\vee_n(z_{\llbracket n \rrbracket}, v_{\llbracket n \rrbracket}) := \prod_{i=1}^n \frac{dy_i}{v_i \hbar} e^{v_i (\mathcal{S}( v_i \hbar \partial_{y_i})-1)x_i} \sum_\Gamma \frac{1}{\vert \mathrm{Aut}(\Gamma)  \vert} \nonumber \\
& \prod_{ \bullet \in W(\Gamma)} \prod_{i \in I(\bullet)}   v_i \hbar \mathcal{S} (v_i \hbar \partial_{y_i})           \frac{\omega^\vee_{\vert I(\bullet) \vert} \left( z_{  I(\bullet)  } \right)}{\prod_{i\in I(\bullet)} d y_i}
\end{align}
with the analogous expansion $\mathbb{W}_n^\vee := \sum_{g=0}^\infty \hbar^{2g-2+n} \mathbb{W}^\vee_{g,n}$. Similar properties hold for the dual differentials $\mathbb{W}_n^\vee$ as the ones for $\mathbb{W}_n$. The inverse relation allows to get from the dual family to the original one, it has exactly the same expression interchanging $u$ and $v$, $x$ and $y$ and the dual $\mathbb{W}^\vee_n(z_{\llbracket n \rrbracket}, v_{\llbracket n \rrbracket})$, that is
\begin{align} \label{transformationdual2}
      \omega_{n} (z_{\llbracket n \rrbracket})= (-1)^n \left( \prod_{i=1}^n \sum_{r=0}^\infty \left( - d_i \frac{1}{d x_i} \right)^r [u_i^r]  \right)    \mathbb{W}^\vee_n(z_{\llbracket n \rrbracket}, v_{\llbracket n \rrbracket})\, .
    \end{align}
This involution $(\omega_{g,n}^\vee)^\vee=\omega_{g,n}$ is highly nontrivial and can be proved purely combinatorially, thus it is independent of whether $dx$ and $dy$ are meromorphic or admit essential singularities.
\\

Note that until now, we did not assume any recursion satisfied by the differentials, however, when assuming one, for instance Gen-TR, one has the following result 
\begin{proposition}[Thm. 3.6 of \cite{Alexandrov:2024tjo}] \label{ProGenTR} If the family of differentials $\omega_{g,n}$ satisfies the recursion defined by Gen-TR for the initial data $(\Sigma, dx, dy, B, \mathcal{P})$, then the $x$-$y$ dual family of differentials $\omega_{g,n}^\vee$ satisfies automatically the Gen-TR for the dual initial data $(\Sigma, dy, dx, B, \mathcal{P}^\vee)$. As a corollary, one gets the same analogous result when the initial data is restricted to the case of the original TR. 
\end{proposition}
The goal from recalling this is to test the compatibility of our formula with the $x$-$y$ swap relation as well, this will be reformulated in the next subsection. Before this, let us note some remarks. 
\begin{remark}
    Let us first note that for a given non-special point $q \in \Sigma$, the dual differentials $\omega_{g,n}^\vee$ are holomorphic at this point if and only if the differentials $\omega_{g,n}$ are themselves holomorphic. The proof of this statement relies on the definition of the special points and on a filtration of Laurent series which is preserved by the transformation \eqref{transformationdual}. Furthermore, the principal part of the poles of $\omega_{g,n}$ at $q \in \mathcal{P}$ is equivalent to the requirement that the dual differentials $\omega_{g,n}^\vee$ are holomorphic at this same point. In fact, the special points that are poles of the dual differentials are only the $\vee$-key points in the setting of Gen-TR.
\end{remark}

The case when the dual side is trivial includes a wide set of curves, this is essentially the case when the $\vee$-key points are taken to be the empty set (all the special points are taken in $\mathcal{P}$), in this case, computations become easier to perform (we will see that explicitly in the example section). For this, let us note that there exists an explicit closed expression for the differentials. This formula is due to the fact that when the dual side is trivial, the unlabeled vertices of the graphs of Def.~\ref{x-yswapdef} have valency $=2$, in other words, any unlabeled vertex is connected to exactly two labeled vertices. From a combinatorial perspective, the unlabeled 2-valent vertices can be removed, yielding connected graphs on $n$ labeled vertices with multi-edges (these multi-edges are different from the previous ones). In this case, the formula becomes: 
\begin{corollary} \label{xytrivialdual} [\cite{Hock:2023dno,Alexandrov:2022ydc}]  
    In the case where $\mathcal{P}^\vee=\emptyset$, the differentials of Gen-TR can be computed via the following formula
    \begin{align}
        \mathbb{W}^\vee _n(z_{\llbracket n \rrbracket}, v_{\llbracket n \rrbracket})= &\prod_{i=1}^n \left( e^{v_i \left( \mathcal{S}(v_i \hbar \partial_{y_i}) - 1 \right)x_i} 
        \sqrt{\frac{d \hat{z}_i^+}{dz_i} \frac{d\hat{z}_i^-}{dz_i}} dz_i \right)  (-1)^{n-1} \sum_{\sigma \in \mathrm{Cycle}(n)} \prod_{i=1}^n \frac{1}{\hat{z}_i^+ - \hat{z}_{\sigma(i)}^-}\,, \nonumber \\
        \omega_{n} (z_{\llbracket n \rrbracket})=& (-1)^n \left( \prod_{i=1}^n \sum_{r=0}^\infty \left( - d_i \frac{1}{d x_i} \right)^r [v_i^r]  \right)    \mathbb{W}^\vee_n(z_{\llbracket n \rrbracket}, v_{\llbracket n \rrbracket})\, .
    \end{align}
    where $\sigma \in \mathrm{Cycle}(n)$ denotes a permutation consisting of a single cycle of length $n$. We also define 
    \begin{align}
        \hat{z}(z,v) = e^{ \frac{v \hbar}{2}  \partial_y} z = z+ \frac{v}{2y'} \hbar - \frac{v^2 y''}{8 (y')^3} \hbar^2 + \dots, \qquad \hat{z}_i^\pm = \hat{z} (z_i , \pm v_i)\,.
     \end{align}
     Note that above we have used the formula for the inverse duality explicitly given above and we have $x'_i = x'(z_i)$ and $y'_i = y'(z_i)$. Under the additional assumption that for $y(z)=z$, we have $\hat{z}_i^\pm = z_i \pm \frac{\hbar v_i}{2}.$
\end{corollary}
The proof of the above expression is not straightforward, so let us explain some of the ingredients of the explicit formula starting from Def.~\ref{x-yswapdef}; for more details, see \cite{Hock:2023dno,Alexandrov:2022ydc}. By summing over all connected graphs on $n$ labeled vertices with multi-edges, it is possible to incorporate multi-edges into exponential generating series while taking into account the order of the automorphism group. We are left with connected graphs on $n$ labeled vertices. The exponentiation together with the action of
$\mathcal{S}(\hbar u_i\partial_{y_i})\mathcal{S}(\hbar u_j\partial_{y_j}),$
produces $\hat{z}_i^+ - \hat{z}_j^-$ for each edge of a connected graph on $n$ labeled vertices. The formula then relies on passing from connected to disconnected ones, thereby turning the sum over all connected graphs into disconnected ones which can be described by a product over all pairs of labeled vertices. Furthermore, one observes that this expression can be rewritten in a form identical to the determinant of a Cauchy matrix. Writing the determinant as sum over all permutations, we can restrict to the connected part back again. All contributions arising from permutations with more than one cycle vanish, leaving only the sum over the remaining $n$-cycles (interpreted as connected part of a determinant).

\subsection{Extension to ramification points of infinite order}

We now address the main interest of this paper, the problem of an essential singularity. Throughout our discussion, one may consider the case of an exponential singularity, but our construction will cover other types as well since we will rely on the infinite principal part of the Laurent series of these functions, exposing their essential singularities. We note that we require the function holding the essential singularity to admit an expression as an infinite limit of another meromorphic function (with no essential singularities). For instance, the $\sin$ and $ \cos$ functions are non-constant entire functions that are not polynomial and hence exhibit an essential singularity at $\{ \infty \}$. 
For instance with their Laurent series expansions, we consider such functions as well.

As we have seen, Gen-TR is formulated as a sum of residues over a chosen set of key points. When one of these key points becomes an essential singularity (seen as a meromorphic singularity or ramification point of infinite order), the recursion is no longer directly suitable, since residues at essential singularities involve infinitely many coefficients of the Laurent expansion. The purpose of this section is to show that this apparent obstruction is artificial: the global nature of Gen-TR in the compact setting allows one to replace every residue at essential singularities by residues at ordinary meromorphic points. 

The principle behind this is that whenever an essential singularity belongs to the set $\mathcal{P}$, Lemma~\ref{contourdef} enables its contribution to be transferred to the complementary poles of the integrand. Consequently, the recursion never requires evaluating residues at the essential singularity. In other words, we have the following result. 

\begin{theorem} \label{main}
    Let $(\Sigma,dx,dy,B,\mathcal{P})$ be the initial data of Gen-TR, where $dx$ (or $dy$) possesses an isolated essential singularity at a point $q\in \Sigma$. Assume this singularity is chosen to belong to the set of key points $\mathcal{P}$. The multi-differentials are well-defined by Lemma~\ref{contourdef}. The recursive residue formula involves only ordinary meromorphic points with Laurent series admitting finite polar part. In particular, no residue at the essential singularity needs to be evaluated.
\end{theorem}
\begin{proof}
Let us assume that the essential contribution is given in $dx$, this means that near the point $q\in \Sigma$ the function $dx$ admits a Laurent series of infinite order polar part. Note now that the same property holds for the multi-differentials $\omega_{g,n}$ due to their polynomial combination expression in terms of the differentials $dx,dy, \omega_{g',n'}$ and their derivatives. This property provides the global nature of these differentials. We choose the set of key points to be the set of singularities given by the infinite convergent sequence contributing to the essential singularity. The recursion of Gen-TR depends on the residues of the global differentials $\bar{\omega}_{g,n}$, since the latter is globally meromorphic away from the prescribed singularities (due to our choice of the set of key points), the global residue theorem allows the contour surrounding an essential singularity to be continuously deformed into contours enclosing all the remaining poles. This contour deformation is given precisely by Lemma~\ref{contourdef}. The case of an essential singularity in the function $y$ could be equivalently treated via the $x$-$y$ dual formula interchanging the roles of $x$ and $y$. The reason is that the $x$-$y$ duality is purely combinatorial and therefore holds independently of whether $dx$ or $dy$ admits an essential singularity.
\end{proof}

\begin{remark}
    One may wonder if the above theorem is valid in cases where both functions $dx$ and $dy$ admit an essential singularity and the answer is in principle positive. The proof of this statement is similar to the proof of the theorem: one chooses the set of key points $\mathcal{P}$ to have all the essential contribution despite the fact that this could be different than the canonical choice of key and $\vee$-key points, this also includes the case where the point is an essential singularity of both $dx$ and $dy$.
\end{remark}

We note that this theorem is not merely formal, in fact, previous attempts to accommodate essential contribution relied on limits of finite ramification points (meromorphic contribution). On the contrary, in this setup, one does not need such an assumption as long as the essential contribution is located and taken to be in the set of key points $\mathcal{P}$. In other words, the recursion itself never needs to see the essential contribution after moving the contour. We will show this on an example where no limit of a sequence of meromorphic spectral curves is admitted in Sec.~\ref{Sec3}. 

Let us show this on a natural example 
\begin{example}
    Take the spectral curve on $\Sigma=\mathbb{P}^1$ admitting the following parameterization 
    \begin{align}
        x = z e^z, \qquad y=z
    \end{align}
    One has an exponential (essential) singularity, one expresses the exponential as an infinite sequence $e^z = \underset{r \to \infty}{\lim} (1+ \frac{z}{r})^r$ allowing to see that one of the special points is located at $z=-r$ of infinite order in the limit. Note that in this case, the set of multi-differentials computed for each $r$ denoted $\omega^r_{g,n}$ converges in the limit. Taking this point as the only key point $\mathcal{P}=\{ -r\}$, its complement, the set of $\vee$-key is $\mathcal{P}^\vee=\{\frac{-r}{1+r}\}$. Thus, Lemma~\ref{contourdef} allows us to move the contour of integration getting then the residue contribution from $\mathcal{P}$ through the residue at $\mathcal{P}^\vee$ and $\{z,z_{\llbracket n  \rrbracket}\}$. We detail a similar example in the next section. \\

    The relation between $x$ and $y$ is
\begin{align*}
    x=ye^y.
\end{align*}
Interestingly, this is exactly the same relation as for simple Hurwitz numbers \cite{Bouchard2008} (possibly after a simple symplectic shift). However, the parametrization differs by the change of variable $z \mapsto \log z$. All multi-differentials are different and we do not know if there is a relation.
\end{example}

When the set of $\vee$-key points is chosen to be the empty set (which allows to get a more simplified residue), the dual side becomes automatically trivial from Prop.~\ref{ProGenTR}. Let us stretch the above example to include a wider range of spectral curves. 
\begin{proposition}[Collected essential contribution]\label{Essentialcontribution}
   Let $\Sigma= \mathbb{P}^1$ and consider a non-constant rational function $g(z)$ and the following family of spectral curves
    \begin{align}
        x = e^{g(z)}, \qquad y= z\,.
    \end{align}
    The differentials $\omega_{g,n}$ of the generalized topological recursion can be computed through Lemma~\ref{contourdef} by taking the set of key points to consist at least of $\{\, \text{poles of $g(z)$} \, \}$. The set of zeroes of $g'(z)$ can be chosen arbitrarily as key or $\vee$-key points. 
\end{proposition}
\begin{proof}
    The first step in the proof relies on establishing a limit for the function $x$ (and therefore to $dx$ as well). In fact, we would like to have 
    \begin{align}
        e^{g(z)} = \underset{n \to \infty}{\lim} \left( 1 + \frac{g(z)}{n} \right)^n, \qquad \text{where $g(z)$ is a rational function}
    \end{align}
    For this, one needs to check the convergence of the expression, since we are in the compact case, consider a compact subset $U$ and a local coordinate on $w \in U$, away from the set of poles of $g(z)$ one has a uniform convergence 
    \begin{align}
        \left( 1 + \frac{w}{n} \right)^n \to e^w, \quad \text{for  $\vert w \vert \leq M$}
    \end{align}
    with $M$ being the supremum of the constant function $g(z) \in U$. 
    The true problem occurs near a pole $a$ of the function $g(z)$, for this one has to take a punctured neighborhood of that pole. In other words, $g(z)$ is always holomorphic at the punctured disc $D^*(a,r)= \{ 0 < \vert z -a \vert <r \}$. Then one applies the same reasoning as the holomorphic part. This allows us to get the required exponential limit. Due to the limit formula, the set of poles of the function $g(z)$ will be the set of essential singularities, this is seen from
    \begin{align}
        \left( 1 + \frac{g(z)}{n} \right)^n = \left( \frac{g(z)}{n} \right)^n \left(1 + \frac{n}{g(z)} \right)^n
    \end{align}
    For fixed $n$, the second factor tends to $1$ near a singularity $a$ while the first term controls it with an exponent $n$. These poles together with the zeros of $g'(z)$ constitute the set of special points of Def.~\ref{def:specialpoint}. The main idea now is to choose the poles of $g(z)$ as key points. One thus has still the freedom to choose the zeroes of $g'(z)$ in $\mathcal{P}$ or $\mathcal{P}^\vee$. A contour deformation as in Lemma~\ref{contourdef} gives us the $\omega_{g,n}$ as a residue computation of points that are meromorphic. 
\end{proof}

\begin{remark}
    Note that the above proposition naturally generalizes to the case of non-trivial differentials $dy$ which admits zeroes and poles (for instance $y$ given by a rational function of $z$). This covers also the case where an essential singularity is at the same time a pole or a simple zero of the differential $dy$. One just needs to pay attention to the choice of key and $\vee$-key points.  
\end{remark}
The above result extends to other situations as well relying on the fact that a certain sequence $\phi_n$ converges locally uniformly on a compact subset on which the function $g$ is holomorphic. So that the limit is really an instance of $\phi_n \circ g \to e \circ g$. The result of the above proposition is applied only for the genus $0$ case due to nice rational parameterizations of the curve. However, our residue formula applies to any spectral curve without the need for an explicit parameterization and more importantly to higher genus curves. For instance, the idea of a convergent sequence leading to the contribution of an essential singularity is still valid independently from the underlying compact surface $\Sigma$. 
\begin{example}
   An important genus-one curve in the context of topological recursion was considered in \cite{Iwaki:2019zeq}. The curve is given by
\begin{align*}
    y^2=x^3+ax+b,
\end{align*}
and was shown to provide a formal solution of the Painlevé I equation. It can be parametrized in terms of the Weierstraß $\wp$-function by
\[
x=\wp(z), \qquad y=\wp'(z).
\]

To obtain an essential singularity, one may consider the limit $r\to \infty$ of
\[
x_r=\left(1+\frac{\wp(z)}{r}\right)^r,
\]
which converges to
\[
x_\infty=e^{\wp(z)},
\]
while keeping $y=\wp'(z)$. The Bergman kernel can also be expressed in terms of the Weierstraß function; see \cite{Iwaki:2019zeq}. This curve falls into the class of spectral curves that can now be treated using  Thm.~\ref{main}.

The exponential singularity is located at $z=0$, which is simultaneously a pole of $\wp$ and $\wp'$. The formulas of Lemma~\ref{Lemmaeq1} remain perfectly valid. The additional special points are the zeros of $\wp'(z)$ and $\wp''(z)$. We are not aware of any enumerative interpretation of the multi-differentials $\omega_{g,n}$ associated with this curve.
\end{example}

Our final remark on this matter is that this residue deformation is indeed compatible with the $x$-$y$ duality considered in the previous section. 
\begin{corollary}
The $x$–$y$ swap extends to the case in which $dx$  and $dy$ admit essential singularities.
\end{corollary}
Differentials and their duals ensure that the initial data of the dual differentials are defined on the set of $\vee$-key points. This equivalently offers an alternative approach to dealing with essential singularities through the consideration of the $x$–$y$ dual differentials. Since the $x$–$y$ duality for an arbitrary system of multi-differentials holds -- and can in fact be proven by combinatorial methods independent of the specific forms of $dx$ and $dy$ -- this result follows.

While this strategy remains widely used and preferred when the dual side is trivial, there are some limitations especially when the dual side is complicated. The contour deformation approach adds a residue contribution at the points $\{z,z_i \}$. However, it does not rely on the triviality of the dual side. It therefore offers an alternative resolution to the problem. We will now turn our attention to the examples.

\section{Examples} \label{Sec3}
We present a sequence of examples illustrating our construction from different perspectives and with increasing complexity.

Before turning to examples of spectral curves with essentially singular ramification points, we first discuss an example of so-called logarithmic topological recursion (Log-TR). This example can also be interpreted as a spectral curve with an exponential singularity. Whether one views the singularity as logarithmic or exponential is determined by the choice of Bergman kernel (equivalently by the projection) as we will show.

For curiosity, we also include a transcendental example involving a spectral curve related to the Riemann \(\zeta\)-function with a trivial dual side.

Finally, we discuss the so-called Mirzakhani curve, which computes Weil--Petersson volumes and naturally exhibits an essential singular point on the dual side. After applying a suitable symplectic transformation, the dual side of the Mirzakhani curve can be simplified from differentials possessing infinitely many poles and an essential singularity to differentials with just a single essential singular point. Those differentials admit a very interesting representation, appearing to have poles at the ramification points of $x$ and along the diagonal, which are apparent, thus regular points.

\subsection{Log-TR example}\label{sec:logTR}
Let $\Sigma=\mathbb{P}^1$ and let the two meromorphic differentials be defined by
\begin{align}\label{logTRpara}
    dx=dz\,\qquad dy=\frac{dz}{z}.
\end{align}
The differential $dy$ possesses a logarithmic singularity corresponding to $y=\log z$. As a global curve, $x$ and $y$ therefore satisfy the equation
\begin{align}
    e^y-x=0.
\end{align}
Furthermore, in the parametrization \eqref{logTRpara} in terms of the global parameter $z$, the Bergman kernel takes the form
\begin{align*}
    B(z_1,z_2)=\frac{dz_1\,dz_2}{(z_1-z_2)^2}.
\end{align*}
To specify the full spectral curve in Gen-TR, we have to provide the set of key-points, which in this case is $\mathcal{P}=\{0\}$. Note that the only special point is at $z=0$, implying that the dual $\vee$-key is empty, $\mathcal{P}^\vee=\emptyset$.

This is a very specific setting in which we consider Gen-TR. We essentially have the following situation: $dx$ and $dy$ are meromorphic, and the set of key-points is taken to be the union of the ramification points of $x$ (which we denote by $\{q_i\}$) and the points where $dy$ has residues but $dx$ is regular (which we denote by $\{a_j\}$), that is, $\mathcal{P}=\{q_i\}\cup \{a_j\}$. This is exactly the situation in which the recursion of Gen-TR is equivalent to the so-called Log-TR \cite{Hock:2023dno,Alexandrov:2023tgl,Alexandrov:2024tjo,Hock:2026qrp}. Thus, Log-TR is one particular realization of Gen-TR from Def.~\ref{Gen-TRDef} in the above setting. The differentials in Log-TR are defined by $\omega_{0,1}:=ydx$ and $\omega_{0,2}:=B$, and for any $g\geq 0$ and $n\geq 1$ such that $2g+n-2>0$ locally by
\begin{align}
        \label{LogTRDef} \omega_{g,n}(z_1,\dots,z_n)=&\frac{1}{2}\sum_{i}\Res_{z\to q_i}\frac{\int_{z}^{\sigma_i(z)} \omega_{0,2}(z_1,.)}{\omega_{0,1}(\sigma_i(z))-\omega_{0,1}(z)}\Big(\omega_{g-1,n+1}(z,\sigma_i(z),z_2,\dots,z_n)\\\nonumber
    &+\sum_{\substack{g_1+g_2=g\\I_1\sqcup I_2=\{2,\dots, n\} \\(g_i,|I_i|)\neq (0,0)}} \omega_{g_1,|I_1|+1}(z,z_{I_1}) \omega_{g_2,|I_2|+1}(\sigma_i(z),z_{I_2}) \Big)\\\nonumber
    &+\delta_{n,1}\sum_{j}\Res_{z\to a_j}\left(\int_{a_j}^z\omega_{0,2}(z_1,.)\right)dx(z)[\hbar^{2g}]\left(\frac{y_{a_j}}{\mathcal{S}(y_{a_j}^{-1}\hbar \partial_x)}\ln(z-a_j) \right),
\end{align}
where $y_j=\Res_{z\to a_j}dy$. As follows from the definition of Log-TR, at a ramification point of $x$ the definition coincides with that of the original TR, whereas at the points $\{a_j\}$ the evaluation is different, but is in fact given explicitly.

For the spectral curve \eqref{logTRpara}, we have more explicitly, with $\mathcal{P}=\{q_i\}\cup \{a_j\}=\emptyset\cup \{0\}$,
\begin{align*}
    \omega_{g,1}(z)=dz[\hbar^{2g}]\left(\frac{1}{\mathcal{S}(\hbar \partial_z)}\ln(z) \right)=-\frac{(2g-1)!dz}{z^{2g}}[\hbar^{2g}]\frac{1}{\mathcal{S}(\hbar)},
\end{align*}
and $\omega_{g,n}=\delta_{n,2}\delta_{g,0}B$ for $n>1$.

As the origin of Log-TR comes from the $x$-$y$ duality, the statement is that the dual curve $(\Sigma,dy,dx,\mathcal{P}^\vee,B)$ is also governed by Log-TR, which implies that $\mathcal{P}^\vee$ consists of the ramification points of $y$ and the points where $dx$ has a residue and $dy$ is regular. For the curve \eqref{logTRpara}, this implies that all dual differentials $\omega_{g,n}^\vee$ vanish for $2g-2+n>0$, since $\mathcal{P}^\vee=\emptyset$.

\subsection{Exponential ramification point with trivial dual side}\label{sec.exampleez}
Let us consider the following example of a spectral curve $(\Sigma=\mathbb{P}^1,dx=e^zdz,dy=dz,B=\frac{dz_1\, dz_2}{(z_1-z_2)^2},\mathcal{P}=\{\infty\})$. Most importantly, note that the relation between $x=e^z$ and $y=z$ is the same as for the curve \eqref{logTRpara},
\begin{align*}
    e^y-x=0.
\end{align*}
However, the main difference comes from the choice of the Bergman kernel, or rather from the choice of the global coordinate $z$ (projection). This in turn determines where the special points are located. In the previous example, there was only one special point at the origin, whereas now the special point is located at infinity.

One might argue that this choice of coordinate is simply a poor choice and that the formalism of Gen-TR is not applicable. However, this curve can be approximated by the sequence of curves
\begin{align}\label{Crcurve}
    C_r:=(\Sigma=\mathbb{P}^1,dx_r=(1+\frac{z}{r})^{r-1}dz,dy=dz,B=\frac{dz_1\, dz_2}{(z_1-z_2)^2},\mathcal{P}=\{-r\}).
\end{align}
This is a perfectly well-defined spectral curve with a single ramification point of order $r$ at $z=-r$. For the spectral curve $C_r$, Def.~\ref{Gen-TRDef} applies directly and computes the same differentials as the higher-order TR defined by Bouchard and Eynard \cite{Bouchard:2012yg}, as follows from \cite[Thm.~5.6]{Alexandrov:2024tjo}. Let $\omega_{g,n}^r$ denote the corresponding differentials associated with the spectral curve $C_r$, computed either by Gen-TR or, equivalently, by BE-TR. Then the limit $\lim_{r\to \infty}\omega_{g,n}^r$ converges to the differentials with essential singularities defined in our main result Thm.~\ref{main} and Lemma~\ref{contourdef}.

To our understanding, this was in fact the original motivation behind the article \cite{Bouchard:2023yau}, where a limit of BE-TR with infinite-order ramification points was studied, but only for specific curves for which the resulting differentials remain meromorphic, since $\omega_{0,1}$ is required to be meromorphic there; see \cite[Def.~3.6]{Bouchard:2023yau}. Since the spectral curve \eqref{Crcurve} does not satisfy this condition in the limit $r\to\infty$, the approach of \textit{op.~cit.} is not applicable in this case, whereas our approach remains well-defined.

Since we have several ways of computing $\omega_{g,n}^r$, either via Gen-TR, BE-TR, or the $x$-$y$ duality, we can choose our preferred one. Let us observe that the dual curve
\[
C_r^\vee:=(\Sigma=\mathbb{P}^1,dy=dz,dx_r=(1+\frac{z}{r})^{r-1}dz,B=\frac{dz_1\, dz_2}{(z_1-z_2)^2},\mathcal{P}^\vee=\emptyset)
\]
is trivial in the sense that all stable $(\omega_{g,n}^r)^\vee=0$. Thus, the $x$-$y$ duality formula of Cor.~\ref{xytrivialdual} reads
\label{xyformulatrivialdual}
\begin{align}\label{xyfiniter}
    \omega_{g,n}^r=&[\hbar^{2g+n-2}]\prod_{i=1}^n\sum_{m_i\geq 0}\big(-d_i\frac{1}{d(x_r)_i}\big)^{m_i}[u_i^{m_i}](-dy_i)\frac{e^{\frac{1}{\hbar}\int_{y_i-\hbar u_i/2}^{y_i+\hbar u_i/2}(x_r)_idy_i-u_i(x_r)_i}}{u_i\hbar}\\\nonumber
    &\times \sum_{\sigma\in \, \mathrm{Cycle}(n) }\prod_{j=1}^n\frac{u_ju_{\sigma(j)}\hbar ^2}{(y_j-y_{\sigma(j)})^2-\frac{\hbar^2}{4}(u_j+u_{\sigma(j)})^2},
\end{align}
with $(x_r)_i=x_r(z_i)=(1+\frac{z_i}{r})^{r}$ and $y_i=z_i$.

For instance, it is easy to derive
\begin{align*}
    \omega^r_{0,3}(z_1,z_2,z_3)=&d_1d_2d_3\bigg[\frac{1}{(z_1-z_2)(z_1-z_3)x_r'(z_1)}+\frac{1}{(z_2-z_1)(z_2-z_3)x_r'(z_2)}\\
    &+\frac{1}{(z_3-z_1)(z_3-z_2)x_r'(z_3)}\bigg],\\
    \omega_{1,1}^r(z)=&-\frac{(r-1) (r+1)}{24 (r+z)^3 x'_r(z)}dz,\\
    \omega_{2,1}^r(z)=&\frac{(r-3) (r-1) (r+1) (r+3) (2 r+1) (2 r+3)}{640 r^3 (r+z)^4x_r(z)^3}.
\end{align*}
Further differentials become cumbersome very quickly. It is easy to see that the differentials $\omega_{g,n}^r$ have a well-defined limit as $r\to\infty$. The same holds for the formula \eqref{xyfiniter}. The limit $r\to\infty$ is
\begin{align}\label{omgnxezyz}
    \omega^\infty_{g,n}(I)=&[\hbar^{2g+n-2}]\prod_{i=1}^n\sum_{m_i\geq 0}\big(-d_i\frac{1}{dx_i}\big)^{m_i}[u_i^{m_i}](-dy_i)\frac{e^{\frac{e^{y_i+\hbar u_i/2}-e^{y_i-\hbar u_i/2}}{\hbar}-u_ie^{y_i}}}{u_i\hbar}\\
    &\times \sum_{\sigma\in \,\mathrm{Cycle}(n) }\prod_{j=1}^n\frac{u_ju_{\sigma(j)}\hbar ^2}{(y_j-y_{\sigma(j)})^2-\frac{\hbar^2}{4}(u_j+u_{\sigma(j)})^2},
\end{align}
with $x_i=e^{y_i}$.

Interestingly, it turns out that all $\omega^\infty_{g,1}$ vanish in the limit except for $\omega^\infty_{0,1}$. This is due to the following fact:
\begin{align*}
    \omega^\infty_{g,1}(z)=&[\hbar^{2g-1}]\sum_{m\geq 0}\big(-d\frac{1}{dx}\big)^{m}[u^{m}](-dy)\frac{e^{\frac{e^{y+\hbar u/2}-e^{y-\hbar u/2}}{\hbar}-ue^{y}}}{u\hbar}\\
    =&\sum_{m\geq 0}\big(-d\frac{1}{dx}\big)^{m}[u^{m}](-dy)\sum_{k=2g}^{3g-1}u^ka_{g,k}e^{y(k-2g)}
    =0.
\end{align*}
We have expanded the exponential and extracted the coefficient $[\hbar^{2g-1}]$, where $a_{g,k}$ are certain coefficients independent of $x$. Replacing $e^{y(k-2g)}dy=x^{k-2g}dy$ in the second-to-last line, we see that the action of $d\frac{1}{dx}$ reduces the exponent by one, since $dx=x\,dy$. If we now act $m=k$ times, the result vanishes for $g>0$.

However, the differentials $\omega^\infty_{g,n}$ with $n>1$ do not vanish in general. For instance,
\begin{align*}
    \omega_{0,3}^\infty(z_1,z_2,z_3)=&d_1d_2d_3\bigg[\frac{e^{-z_1}}{(z_1-z_2)(z_1-z_3)}+\frac{e^{-z_2}}{(z_2-z_1)(z_2-z_3)}+\frac{e^{-z_3}}{(z_3-z_1)(z_3-z_2)}\bigg],
\end{align*}
which is regular along the diagonal.\\

We can now compare the differentials of Sec.~\ref{sec:logTR} and \eqref{omgnxezyz}, which are both generated by Gen-TR for the same relation between $x$ and $y$, but with Bergman kernels corresponding to two different global coordinates that are not related by a holomorphic transformation. In both cases, the dual side is trivial. For the differentials of Sec.~\ref{sec:logTR}, we find that all $\omega_{g,1}$ are non-vanishing, whereas for the spectral curve considered here the situation is rather the opposite: all $\omega^\infty_{g,1}$ vanish, but for $n>1$ the differentials are non-vanishing in general. Furthermore, the $\omega^\infty_{g,n}$ are generated as the limit of a convergent sequence of rational spectral curves.

\subsection{Riemann $\zeta$-function as spectral curve}\label{sec.riemannzeta}
All examples in the realm of TR or its extensions are constructed from some local behavior of the two algebraic functions $x,y$, or rather the two meromorphic forms $dx,dy$. In the previous example, we have seen that these forms can actually have exponential or essential singular points, giving rise to exponential or essential singularities for all $\omega_{g,n}$. The corresponding global curve can be realized as the vanishing locus of a polynomial in the variables $e^x,e^y$ and/or $x,y$. These examples are still within the known class of curves appearing as mirror curves, Hurwitz numbers, or other familiar examples with possible different projections. However, in principle, one can take $x$ (or $dx$) to be any function (or form) on the complex plane and set $y=z$. Thus, our new approach is not necessarily restricted to meromorphic forms or ones with exponential singularities. Take, for instance, the following example:
\begin{align}\label{zetaspec}
    x=\zeta(z),\qquad y=z,\qquad\Rightarrow \quad x=\zeta(y),
\end{align}
where $\zeta$ denotes the Riemann $\zeta$-function, a transcendental function defined by
\begin{align*}
    \zeta(z)=\sum_{n\geq 1}\frac{1}{n^z},\qquad \text{for $\mathrm{Re}(z)>1$}.
\end{align*}
It admits a unique meromorphic continuation to the whole complex plane. Thus, this function is defined on the entire complex plane and has an essential singularity at infinity. Our approach provides a way to generate the corresponding $\omega_{g,n}$. A natural question is whether, by taking $x$ to be the Riemann $\zeta$-function, some non-trivial number-theoretic information is encoded in the differentials $\omega_{g,n}$.

The Bergman kernel is the standard one on the sphere,
\[
B=\frac{dz_1\,dz_2}{(z_1-z_2)^2}.
\]
The corresponding Gen-TR spectral curve is
\begin{align*}
    (\mathbb{P}^1,d\zeta(z),dz,B,\mathcal{P}),
\end{align*}
with $\mathcal{P}=\{\text{special points}\}$. The set of special points includes infinity and all zeros of $\zeta'(z)$, since $r\geq 2$ and $s=1$ in Def.~\ref{def:specialpoint}. The dual curve has the empty set as its set of dual key-points, which in turn produces vanishing dual correlators $\omega_{g,n}^\vee=0$ for $2g-2+n>0$. One can compute $\omega_{g,n}$ either by the formula of Lemma~\ref{contourdef} or by the $x$-$y$ duality formula with trivial dual side of Cor.~\ref{xytrivialdual}. For instance, one finds
\begin{align*}
    \omega_{0,3}(z_1,z_2,z_3)=&d_1d_2d_3\bigg[\frac{\frac{1}{(z_1-z_2) \zeta '(z_2)}-\frac{1}{(z_1-z_3) \zeta '(z_3)}}{z_3-z_2}+\frac{1}{(z_1-z_2) (z_1-z_3) \zeta '(z_1)}\bigg],\\
     \omega_{1,1}(z)=&d_z\bigg[\frac{\zeta ''(z)^2}{24 \zeta '(z)^3}-\frac{\zeta ^{(3)}(z)}{24 \zeta '(z)^2}\bigg],\\
     \omega_{2,1}(z)=&d_z\bigg[-\frac{\zeta ^{(7)}(z)}{1920 \zeta '(z)^4}+\frac{11 \zeta ^{(4)}(z)^2}{1920 \zeta '(z)^5}-\frac{5 \zeta ^{(3)}(z)^3}{288 \zeta '(z)^6}+\frac{35 \zeta ''(z)^6}{384 \zeta '(z)^9}+\frac{7 \zeta ^{(6)}(z) \zeta ''(z)}{1440 \zeta '(z)^5}\\
     &+\frac{13 \zeta ^{(3)}(z) \zeta ^{(5)}(z)}{1440 \zeta '(z)^5}-\frac{5 \zeta ^{(5)}(z) \zeta ''(z)^2}{192 \zeta '(z)^6}+\frac{19 \zeta ^{(4)}(z) \zeta ''(z)^3}{192 \zeta '(z)^7}-\frac{35 \zeta ^{(3)}(z) \zeta ''(z)^4}{128 \zeta '(z)^8}\\
     &+\frac{55 \zeta ^{(3)}(z)^2 \zeta ''(z)^2}{288 \zeta '(z)^7}-\frac{\zeta ^{(3)}(z) \zeta ^{(4)}(z) \zeta ''(z)}{12 \zeta '(z)^6}\bigg].
\end{align*}
All differentials $\omega_{g,n}$ are symmetric, regular along the diagonal, and have poles only at the points in $\mathcal{P}$, namely the zeros of $\zeta'(z)$ and infinity.

At present, it remains unclear whether these differentials possess any further special properties or are of any number-theoretical interest.

\subsection{The Mirzakhani curve}
One of the most important examples and applications of TR correlators is the so-called Mirzakhani curve. Mirzakhani derived a recursive formula for the Weil--Petersson volumes in \cite{MR2264808}, which was proved in \cite{Eynard:2007fi} to be equivalent to the original EO-TR for a specifically chosen spectral curve. After a straightforward rescaling, this spectral curve can be parametrized by
\begin{align}
    x= z^2 , \qquad y= \mathrm{i}\sin(z), \qquad  \omega_{0,2}=B=\frac{dz_1\,dz_2}{(z_1-z_2)^2}.
\end{align}
In the context of Gen-TR, this corresponds to the spectral curve
\begin{align}\label{mirzakhanicurve}
    (\mathbb{P}^1,dz^2,d\, (\mathrm{i}\sin z),B,\{0\}).
\end{align}
Some examples of these differentials are
\begin{align}\label{diffmirza}
    \omega_{0,3}(z)&=d_1d_2d_3\bigg[\frac{1}{2z_1z_2z_3}\bigg]\\\nonumber
    \omega_{1,1}(z)&=d_z\bigg[\frac{1}{48 z^3}-\frac{1}{96 z}\bigg]\\\nonumber
    \omega_{1,2}(z_1,z_2)&=d_1d_2\bigg[-\frac{1}{96 z_1^3 z_2}+\frac{1}{32 z_1^5 z_2}-\frac{1}{96 z_1 z_2^3}+\frac{1}{96 z_1^3 z_2^3}+\frac{1}{32 z_1 z_2^5}+\frac{1}{256 z_1 z_2}\bigg]\\\nonumber
    \omega_{2,1}(z)&=d_z\bigg[\frac{35}{3072 z^9}-\frac{29}{6144 z^7}+\frac{139}{122880 z^5}-\frac{169}{737280 z^3}+\frac{29}{393216 z}\bigg].
\end{align}
Large collections of these differentials can be found in the literature, together with important connections to integrable systems and physics, such as string theory, SYK models, and JT gravity; see, for instance, \cite{Stanford:2019vob}.

As these differentials are of great importance in different areas, one can ask about the $x$-$y$ dual differentials, which carry exactly the same amount of information but rearranged in a different way. The $x$-$y$ dual curve to \eqref{mirzakhanicurve} reads
\begin{align}
    (\mathbb{P}^1,d\sin z,dz^2,B,\{\tfrac{(2k+1)\pi}{2},\infty\})
\end{align}
for all $k\in \mathbb{Z}$. This means that the corresponding differentials have, on the one hand, infinitely many finite-order poles at $\{\frac{(2k+1)\pi}{2}\}_{k\in \mathbb{Z}}$ and, on the other hand, an essential singularity at $\infty$. This, however, seems to increase the complexity in comparison with the original differentials \eqref{diffmirza}.

A much better approach from the $x$-$y$ dual perspective is to take the following symplectic transformation into account:
\[
d\,(\mathrm{i}\sin z)\to d(\cos(z)+ \mathrm{i}\sin z)=de^z.
\]
The transformation yields the following new spectral curve
\begin{align}\label{mirzakhanicurvee}
    (\mathbb{P}^1,dz^2,d\, e^z,B,\{0\}),
\end{align}
which generates exactly the same differentials $\omega_{g,n}$ as \eqref{mirzakhanicurve}. This is not obvious from the Gen-TR perspective, but since both curves are related in the original TR, the transformation $y\to y+\cos(z)$ is simply a shift of $y$ by a function of $x=z^2$. Looking at Definition \eqref{EOTRtdom}, only the difference $y(z)-y(\sigma_i(z))$ is affected by this transformation. However, this difference is actually invariant because the Galois involution is $\sigma(z)=-z$. Thus,
\[
y(z)+\cos(z)-y(\sigma(z))-\cos(\sigma(z))
=
y(z)-y(\sigma(z)).
\]
The benefit of working with \eqref{mirzakhanicurvee} is that the set of special points is reduced from infinitely many points to just two points. Thus, the dual curve to \eqref{mirzakhanicurvee} reads
\begin{align}\label{mirzakhanicurveedual}
    (\mathbb{P}^1,d\, e^z,dz^2,B,\{\infty\}).
\end{align}
Let us denote the corresponding differentials by $\omega_{g,n}^\vee$, which are related to the Mirzakhani differentials \eqref{diffmirza} via the $x$-$y$ duality formula \eqref{transformationdual}. Let us list some examples:
\begin{align*}
    \omega_{0,3}^\vee(z_1,z_2,z_3)&=\frac{1}{2} d_1d_2d_3\bigg[-\frac{1}{z_1 z_2 z_3}+\frac{e^{-z_2}}{z_2 (z_2-z_1) (z_2-z_3)}\\
    &+\frac{e^{-z_3}}{z_3 (z_3-z_1) (z_3-z_2)}+\frac{e^{-z_1}}{z_1 (z_1-z_2) (z_1-z_3)}\bigg]\\
    \omega_{1,1}^\vee(z)&=d_z\bigg[\frac{e^{-z}}{48 z^3}-\frac{1}{48 z^3}+\frac{e^{-z}}{48 z^2}+\frac{1}{96 z}\bigg]\\
    \omega_{2,1}^\vee(z)&=d_z\bigg[-\frac{21 e^{-2 z}}{64 z^9}+\frac{91 e^{-3 z}}{384 z^9}+\frac{105 e^{-z}}{1024 z^9}-\frac{35}{3072 z^9}-\frac{35 e^{-2 z}}{512 z^8}+\frac{231 e^{-3 z}}{1024 z^8}\\
    &+\frac{131 e^{-3 z}}{1536 z^7}
    +\frac{7 e^{-2 z}}{1536 z^7}-\frac{49 e^{-z}}{3072 z^7}+\frac{29}{6144 z^7}+\frac{5 e^{-3 z}}{384 z^6}+\frac{11 e^{-2 z}}{3072 z^6}+\frac{e^{-2 z}}{2304 z^5}\\
    &-\frac{e^{-3 z}}{2880 z^5}
    +\frac{37 e^{-z}}{36864 z^5}
    -\frac{139}{122880 z^5}-\frac{e^{-3 z}}{3840 z^4}+\frac{169}{737280 z^3}-\frac{29}{393216 z}\bigg].
\end{align*}

We want to emphasize again that all these differentials are regular on the diagonal and at the points in $\mathcal{P}=\{0\}$, even though they are represented as having poles at those points. The reason is that the exponential cannot be brought into a partial fraction decomposition, as is usually done in the original EO-TR with meromorphic forms. The expansion around $z=0$ is, for instance,
\begin{align*}
    \omega_{1,1}^\vee(z)&=\bigg(-\frac{1}{384}+\frac{z}{720}-\frac{z^2}{2304}+\frac{z^3}{10080}+\mathcal{O}(z^4)\bigg)dz\\
    \omega_{2,1}^\vee(z)&=\bigg(\frac{6683}{1238630400}+\frac{31583 z}{10218700800}-\frac{2935 z^2}{2179989504}-\frac{859961 z^3}{398529331200}+\mathcal{O}(z^4)\bigg)dz.
\end{align*}
It is still mysterious what kind of enumerative interpretation these differentials have, as they are related to the Weil--Petersson volumes but wrapped in a non-trivial way.

\section{Further directions} \label{Sec4}
There remain many open questions concerning the geometric and enumerative significance of spectral curves with essential singularities. It is currently unclear whether spectral curves with essential singularities arise naturally from existing enumerative theories or whether they define entirely new classes of invariants. Thus, let us give some possible further direction:
\begin{itemize}
    \item \textit{Enumerative geometry}: The passage from $\psi$-class intersection numbers to their higher $r$-spin counterparts \cite{Chidambaram:2022cqc} led to the extension of EO-TR to BE-TR. Since both formalisms are encompassed by Gen-TR, it is natural to investigate the limit $r\to\infty$. Although the corresponding spectral curves do not converge directly, an appropriate rescaling by $r$ yields the curve discussed in Sec.~\ref{sec.exampleez}. This raises the question of whether this curve admits an enumerative interpretation as a limit of rescaled $r$-spin theories, as a kind of "$\infty$-spin" intersection numbers. Similarly, replacing $y$ by $1/y$ often leads to the $\Theta$-class analogue \cite{Norbury:2017eih}, suggesting further possible directions.
    
    On the other hand, when writing a spectral curve  with exponential singularities as a vanishing locus naturally involves exponential functions. Such functions are known to appear for Hurwitz numbers \cite{Bouchard2008,Bychkov:2020yzy} and in mirror curves associated with the Gromov--Witten theory of toric Calabi--Yau threefolds \cite{Bouchard:2007ys} or of $\mathbb{P}^1$ \cite{Norbury2014}. It is therefore natural to ask whether spectral curves with exponential singularities give rise to new enumerative theories, new classes of intersection numbers, new types of Hurwitz numbers, or new connections between topological recursion and existing enumerative problems.

  \item \textit{Matrix models}: Since topological recursion originated in the study of matrix models and loop equations, it is natural to ask whether spectral curves with essential singularities admit a matrix model realization. If so, what probabilistic or statistical-mechanical interpretation should such models possess? It is also conceivable that they arise as scaling limits of more classical matrix models, in analogy with the limit discussed in this article.

  \item \textit{Quantum spectral curves}: When $x$ and $y$ are meromorphic functions on a compact Riemann surface, the spectral curve can often be represented by an algebraic equation
    	$P(x,y)=0.$
This viewpoint motivated the development of quantization schemes based on topological recursion, in which the classical spectral curve is promoted to a quantum spectral curve, namely a differential or difference operator $\hat{P}$ annihilating a wave function $\psi$ constructed from the multi-differentials \cite{Bouchard:2016obz,Iwaki:2019zeq,Eynard:2021sxg,Banerjee:2025shz} $ \hat{P}\,\psi=0.$

To our knowledge, spectral curves with essential singularities of the type considered in this article have not been studied from the perspective of quantum spectral curves. Nevertheless, they seem to be of similar nature as mirror curves, for which quantization schemes are understood through spectral theory and the TS/ST correspondence \cite{Grassi:2014zfa,Francois:2025wwd}. It is therefore natural to investigate a quantization of/by Gen-TR and, for instance, to understand how essential singularities behave. Such a theory could potentially provide a bridge between Gen-TR and other existing quantization schemes.

    \item \textit{Integrable systems}: Topological recursion is naturally related to a variety of integrable systems, depending on the underlying spectral curve \cite{Iwaki:2019zeq,Marchal:2019bia}. More generally, both TR and Gen-TR give rise to KP tau-functions \cite{Alexandrov:2024hgu,Alexandrov:2025sap,Alexandrov:2024zku}. It is therefore natural to ask whether spectral curves admitting essential singularities are related to known integrable systems beyond the classical examples, or whether they correspond to new reductions, extensions, or classes of solutions of existing integrable hierarchies. Understanding this connection could provide further insight into the role of essential singularities within the broader landscape of integrable systems.

    \item \textit{Transcendental multi-differentials}: As the example of Sec.~\ref{sec.riemannzeta} illustrates, Gen-TR can in principle be applied to spectral curves associated with transcendental functions such as the Riemann $\zeta$-function. The example considered there was included primarily as a curiosity. Nevertheless, it demonstrates that the constructions remain well defined and naturally produce an infinite family of multi-differentials associated with transcendental functions.

It is therefore natural to ask whether these constructions have a deeper mathematical meaning and, in particular, whether they can provide a genuinely interesting connection between topological recursion and number theory. For instance, it would be interesting to understand whether arithmetic properties of transcendental functions can be reflected in the resulting multi-differentials and associated partition functions.

\end{itemize}

\section*{Conflict of Interest Statement}
On behalf of all authors, the corresponding author states that there is no conflict of interest.

\bibliographystyle{halpha-abbrv}
\bibliography{omega.bib}

@article{Bouchard:2012yg,
    author = "Bouchard, Vincent and Eynard, Bertrand",
    title = "{Think globally, compute locally}",
    eprint = "1211.2302",
    archivePrefix = "arXiv",
    primaryClass = "math-ph",
    reportNumber = "IPHT-T12-114",
    doi = "10.1007/JHEP02(2013)143",
    journal = "JHEP",
    volume = "02",
    pages = "143",
    year = "2013"
}

@article{Eynard:2007kz,
    author = "Eynard, Bertrand and Orantin, Nicolas",
    title = "{Invariants of algebraic curves and topological expansion}",
    eprint = "math-ph/0702045",
    archivePrefix = "arXiv",
    reportNumber = "SPHT-07-021",
    doi = "10.4310/CNTP.2007.v1.n2.a4",
    journal = "Commun. Num. Theor. Phys.",
    volume = "1",
    pages = "347--452",
    year = "2007"
}

@article{Borot:2021thu,
    author = {Borot, Ga{\"e}tan and Charbonnier, S{\'e}verin and Garcia-Failde, Elba and Leid, Felix and Shadrin, Sergey},
    title = "{Functional relations for higher-order free cumulants}",
    eprint = "2112.12184",
    archivePrefix = "arXiv",
    primaryClass = "math.OA",
    month = "12",
    year = "2021"
}

@article{Hock:2022wer,
    author = "Hock, Alexander",
    title = "{On the $x$--$y$ Symmetry of Correlators in Topological Recursion via Loop Insertion Operator}",
    eprint = "2201.05357",
    archivePrefix = "arXiv",
    primaryClass = "math-ph",
    doi = "10.1007/s00220-024-05043-1",
    journal = "Commun. Math. Phys.",
    volume = "405",
    number = "7",
    pages = "166",
    year = "2024"
}

@article{Alexandrov:2022ydc,
    author = "A. Alexandrov and B. Bychkov and P. Dunin-Barkowski and M. Kazarian and S. Shadrin",
    title = "{A universal formula for the $x-y$ swap in topological recursion}",
    eprint = "2212.00320",
    archivePrefix = "arXiv",
    primaryClass = "math-ph",
    doi = "10.4171/JEMS/1615",
    year = "2025",
    Journal = {J. Eur. Math. Soc.}
}

@article{Hock:2022pbw,
    author = "Hock, Alexander",
    title = "{A simple formula for the $x$--$y$ symplectic transformation in topological recursion}",
    eprint = "2211.08917",
    archivePrefix = "arXiv",
    primaryClass = "math-ph",
    doi = "10.1016/j.geomphys.2023.105027",
    journal = "J. Geom. Phys.",
    volume = "194",
    pages = "105027",
    year = "2023"
}

@article{Alexandrov:2023tgl,
    author = "Alexandrov, Alexander and Bychkov, Boris and Dunin-Barkowski, Petr and Kazarian, Maxim and Shadrin, Sergey",
    title = "{Log Topological Recursion Through the Prism of $x$--$y$ Swap}",
    eprint = "2312.16950",
    archivePrefix = "arXiv",
    primaryClass = "math-ph",
    doi = "10.1093/imrn/rnae213",
    journal = "Int. Math. Res. Not.",
    volume = "2024",
    number = "21",
    pages = "13461--13487",
    year = "2024"
}

@article{Alexandrov:2024hgu,
    author = "Alexandrov, Alexander and Bychkov, Boris and Dunin-Barkowski, Petr and Kazarian, Maxim and Shadrin, Sergey",
    title = "{Any Topological Recursion on a Rational Spectral Curve is KP Integrable}",
    eprint = "2406.07391",
    archivePrefix = "arXiv",
    primaryClass = "math-ph",
    doi = "10.1007/s00220-026-05566-9",
    journal = "Commun. Math. Phys.",
    volume = "407",
    number = "4",
    pages = "69",
    year = "2026"
}

@article{Alexandrov:2024tjo,
    author = "Alexandrov, Alexander and Bychkov, Boris and Dunin-Barkowski, Petr and Kazarian, Maxim and Shadrin, Sergey",
    title = "{Degenerate and Irregular Topological Recursion}",
    eprint = "2408.02608",
    archivePrefix = "arXiv",
    primaryClass = "math-ph",
    reportNumber = "MPIM-Bonn-2024",
    doi = "10.1007/s00220-025-05274-w",
    journal = "Commun. Math. Phys.",
    volume = "406",
    number = "5",
    pages = "94",
    year = "2025"
}

@article{Alexandrov:2024zku,
    author = "Alexandrov, Alexander and Bychkov, Boris and Dunin-Barkowski, Petr and Kazarian, Maxim and Shadrin, Sergey",
    title = "{KP integrability of non-perturbative differentials}",
    eprint = "2412.18592",
    archivePrefix = "arXiv",
    primaryClass = "math-ph",
    month = "12",
    year = "2024"
}

@article{Alexandrov:2025sap,
    author = "Alexandrov, Alexander and Bychkov, Boris and Dunin-Barkowski, Petr and Kazarian, Maxim and Shadrin, Sergey",
    title = "{Blobbed topological recursion and KP integrability}",
    eprint = "2505.03545",
    archivePrefix = "arXiv",
    primaryClass = "math-ph",
    doi = "10.1007/s00029-026-01135-z",
    journal = "Selecta Math.",
    volume = "32",
    number = "2",
    pages = "25",
    year = "2026"
}

@article{Bouchard:2025rid,
    author = "Bouchard, Vincent and Chidambaram, Nitin K. and Giacchetto, Alessandro and Shadrin, Sergey",
    title = "{Theta classes: generalized topological recursion, integrability and $\mathcal{W}$-constraints}",
    eprint = "2505.11291",
    archivePrefix = "arXiv",
    primaryClass = "math.AG",
    month = "5",
    year = "2025"
}

@article{Hock:2023dno,
    author = "Hock, Alexander",
    title = "{$x$--$y$ duality in topological recursion for exponential variables via quantum dilogarithm}",
    eprint = "2311.11761",
    archivePrefix = "arXiv",
    primaryClass = "math-ph",
    doi = "10.21468/SciPostPhys.17.2.065",
    journal = "SciPost Phys.",
    volume = "17",
    number = "2",
    pages = "065",
    year = "2024"
}

@article{Eynard:2007fi,
    author = "Eynard, Bertrand and Orantin, Nicolas",
    title = "{Weil-Petersson volume of moduli spaces, Mirzakhani's recursion and matrix models}",
    eprint = "0705.3600",
    archivePrefix = "arXiv",
    primaryClass = "math-ph",
    reportNumber = "SPHT-T07-065",
    month = "5",
    year = "2007"
}

@article{Stanford:2019vob,
    author = "Stanford, Douglas and Witten, Edward",
    title = "{JT gravity and the ensembles of random matrix theory}",
    eprint = "1907.03363",
    archivePrefix = "arXiv",
    primaryClass = "hep-th",
    doi = "10.4310/ATMP.2020.v24.n6.a4",
    journal = "Adv. Theor. Math. Phys.",
    volume = "24",
    number = "6",
    pages = "1475--1680",
    year = "2020"
}

@article {MR2264808,
    AUTHOR = {Mirzakhani, Maryam},
     TITLE = {Simple geodesics and {W}eil-{P}etersson volumes of moduli
              spaces of bordered {R}iemann surfaces},
   JOURNAL = {Invent. Math.},
  FJOURNAL = {Inventiones Mathematicae},
    VOLUME = {167},
      YEAR = {2007},
    NUMBER = {1},
     PAGES = {179--222},
      ISSN = {0020-9910,1432-1297},
   MRCLASS = {32G15 (14H15)},
  MRNUMBER = {2264808},
MRREVIEWER = {Hsian-Hua\ Tseng},
       DOI = {10.1007/s00222-006-0013-2},
       URL = {https://doi.org/10.1007/s00222-006-0013-2},
}

@article{Bouchard:2007ys,
    author = "Bouchard, Vincent and Klemm, Albrecht and Mari{\~{n}}o, Marcos and Pasquetti, Sara",
    title = "{Remodeling the B-model}",
    eprint = "0709.1453",
    archivePrefix = "arXiv",
    primaryClass = "hep-th",
    reportNumber = "BONN-TH-2007-07, CERN-PH-TH-2007-153, NEIP-07-03",
    doi = "10.1007/s00220-008-0620-4",
    journal = "Commun. Math. Phys.",
    volume = "287",
    pages = "117--178",
    year = "2009"
}

@article{Borot:2023wik,
    author = {Borot, Ga{\"e}tan and Bouchard, Vincent and Chidambaram, Nitin Kumar and Kramer, Reinier and Shadrin, Sergey},
    title = "{Taking limits in topological recursion}",
    eprint = "2309.01654",
    archivePrefix = "arXiv",
    primaryClass = "math.AG",
    doi = "10.1112/jlms.70286",
    journal = "J. Lond. Math. Soc.",
    volume = "112",
    number = "3",
    pages = "e70286",
    year = "2025"
}

@article{Eynard:2012nj,
    author = "Eynard, Bertrand and Orantin, Nicolas",
    title = "{Computation of Open Gromov{\textendash}Witten Invariants for Toric Calabi{\textendash}Yau 3-Folds by Topological Recursion, a Proof of the BKMP Conjecture}",
    eprint = "1205.1103",
    archivePrefix = "arXiv",
    primaryClass = "math-ph",
    reportNumber = "IPHT-T12-030",
    doi = "10.1007/s00220-015-2361-5",
    journal = "Commun. Math. Phys.",
    volume = "337",
    number = "2",
    pages = "483--567",
    year = "2015"
}

@article{Iwaki:2019zeq,
    author = "Iwaki, Kohei",
    title = "{2-Parameter $\tau $-Function for the First Painlev{\'e} Equation: Topological Recursion and Direct Monodromy Problem via Exact WKB Analysis}",
    eprint = "1902.06439",
    archivePrefix = "arXiv",
    primaryClass = "math-ph",
    doi = "10.1007/s00220-020-03769-2",
    journal = "Commun. Math. Phys.",
    volume = "377",
    number = "2",
    pages = "1047--1098",
    year = "2020"
}

@article{Do:2014ncn,
    author = "Do, Norman and Norbury, Paul",
    title = "{Topological recursion for irregular spectral curves}",
    eprint = "1412.8334",
    archivePrefix = "arXiv",
    primaryClass = "math.GT",
    doi = "10.1112/jlms.12112",
    journal = "J. Lond. Math. Soc.",
    volume = "97",
    number = "3",
    pages = "398--426",
    year = "2018"
}

@article{Dubrovin:1994hc,
    author = "Dubrovin, Boris",
    editor = "Francaviglia, M. and Greco, S.",
    title = "{Geometry of 2-D topological field theories}",
    eprint = "hep-th/9407018",
    archivePrefix = "arXiv",
    reportNumber = "SISSA-89-94-FM",
    doi = "10.1007/BFb0094793",
    journal = "Lect. Notes Math.",
    volume = "1620",
    pages = "120--348",
    year = "1996"
}

@article{Marchal:2019bia,
    author = "Marchal, Olivier and Orantin, Nicolas",
    title = "{Quantization of hyper-elliptic curves from isomonodromic systems and topological recursion}",
    eprint = "1911.07739",
    archivePrefix = "arXiv",
    primaryClass = "math-ph",
    doi = "10.1016/j.geomphys.2021.104407",
    journal = "J. Geom. Phys.",
    volume = "171",
    pages = "104407",
    year = "2022"
}

@article{Eynard:2021sxg,
    author = "Eynard, Bertrand and Garcia-Failde, Elba and Marchal, Olivier and Orantin, Nicolas",
    title = "{Quantization of Classical Spectral Curves via Topological Recursion}",
    eprint = "2106.04339",
    archivePrefix = "arXiv",
    primaryClass = "math-ph",
    doi = "10.1007/s00220-024-04997-6",
    journal = "Commun. Math. Phys.",
    volume = "405",
    number = "5",
    pages = "116",
    year = "2024"
}

@article{Bouchard:2016obz,
    author = "Bouchard, Vincent and Eynard, Bertrand",
    title = "{Reconstructing WKB from topological recursion}",
    eprint = "1606.04498",
    archivePrefix = "arXiv",
    primaryClass = "math-ph",
    reportNumber = "IPHT:T16-056, CRM-3354(2016)",
    doi = "10.5802/jep.58",
     JOURNAL = {J. \'Ec. polytech. Math.},
  FJOURNAL = {Journal de l'\'Ecole polytechnique. Math\'ematiques},
    VOLUME = {4},
      YEAR = {2017},
     PAGES = {845--908},
}

@article{Eynard:2011ga,
    author = "Eynard, B.",
    title = "{Invariants of spectral curves and intersection theory of moduli spaces of complex curves}",
    eprint = "1110.2949",
    archivePrefix = "arXiv",
    primaryClass = "math-ph",
    reportNumber = "IPHT-T11-200",
    doi = "10.4310/CNTP.2014.v8.n3.a4",
    journal = "Commun. Num. Theor. Phys.",
    volume = "8",
    pages = "541--588",
    year = "2014"
}

@article{Banerjee:2025shz,
    author = "Banerjee, Sibasish and Hock, Alexander",
    title = "{Quantum curve for strip geometries, topological recursion and open GW/DT invariants}",
    eprint = "2510.07146",
    archivePrefix = "arXiv",
    primaryClass = "math-ph",
    doi = "10.1007/s11005-026-02059-7",
    journal = "Lett. Math. Phys.",
    volume = "116",
    number = "1",
    pages = "22",
    year = "2026"
}

@article{Norbury:2017eih,
    author = "Norbury, Paul",
    title = "{A new cohomology class on the moduli space of curves}",
    eprint = "1712.03662",
    archivePrefix = "arXiv",
    primaryClass = "math.AG",
    doi = "10.2140/gt.2023.27.2695",
    journal = "Geom. Topol.",
    volume = "27",
    number = "7",
    pages = "2695--2761",
    year = "2023"
}

@article{Grassi:2014zfa,
    author = "Grassi, Alba and Hatsuda, Yasuyuki and Marino, Marcos",
    title = "{Topological Strings from Quantum Mechanics}",
    eprint = "1410.3382",
    archivePrefix = "arXiv",
    primaryClass = "hep-th",
    reportNumber = "DESY-14-181",
    doi = "10.1007/s00023-016-0479-4",
    journal = "Annales Henri Poincare",
    volume = "17",
    number = "11",
    pages = "3177--3235",
    year = "2016"
}

@article{Francois:2025wwd,
    author = "Fran{\c{c}}ois, Matijn and Grassi, Alba",
    title = "{On the Open TS/ST Correspondence}",
    eprint = "2503.21762",
    archivePrefix = "arXiv",
    primaryClass = "hep-th",
    reportNumber = "CERN-TH-2025-049",
    doi = "10.1007/s00220-026-05608-2",
    journal = "Commun. Math. Phys.",
    volume = "407",
    number = "7",
    pages = "146",
    year = "2026"
}

@article{Chidambaram:2022cqc,
    author = "Chidambaram, Nitin Kumar and Garcia-Failde, Elba and Giacchetto, Alessandro",
    title = "{Relations on {\textbackslash}overline{{\textbackslash}mathcal{M}}{\_}{g,n} and the negative r-spin Witten conjecture}",
    eprint = "2205.15621",
    archivePrefix = "arXiv",
    primaryClass = "math.AG",
    reportNumber = "MPIM-Bonn-2022",
    doi = "10.1007/s00222-025-01351-y",
    journal = "Invent. Math.",
    volume = "241",
    number = "3",
    pages = "929--997",
    year = "2025"
}

@article{Dunin-Barkowski:2019iog,
    author = "Dunin-Barkowski, Petr and Kramer, Reinier and Popolitov, Alexandr and Shadrin, Sergey",
    title = "{Loop equations and a proof of Zvonkine's $qr$-ELSV formula}",
    eprint = "1905.04524",
    archivePrefix = "arXiv",
    primaryClass = "math.AG",
    doi = "10.24033/asens.2553",
     JOURNAL = {Ann. Sci. \'Ec. Norm. Sup\'er. (4)},
  FJOURNAL = {Annales Scientifiques de l'\'Ecole Normale Sup\'erieure.
              Quatri\`eme S\'erie},
    VOLUME = {56},
      YEAR = {2023},
    NUMBER = {4},
     PAGES = {1199--1229},
}

@article{Bychkov:2020ujd,
    author = "Bychkov, Boris and Dunin-Barkowski, Petr and Kazarian, Maxim and Shadrin, Sergey",
    title = "{Explicit closed algebraic formulas for Orlov-Scherbin $n$-point functions}",
    eprint = "2008.13123",
    archivePrefix = "arXiv",
    primaryClass = "math.CO",
    doi = "10.5802/jep.202",
    month = "8",
    year = "2020"
}

@article{Bychkov:2020yzy,
    author = "Bychkov, Boris and Dunin-Barkowski, Petr and Kazarian, Maxim and Shadrin, Sergey",
    title = "{Topological recursion for Kadomtsev{\textendash}Petviashvili tau functions of hypergeometric type}",
    eprint = "2012.14723",
    archivePrefix = "arXiv",
    primaryClass = "math-ph",
    doi = "10.1112/jlms.12946",
    journal = "J. Lond. Math. Soc.",
    volume = "109",
    number = "6",
    pages = "e12946",
    year = "2024"
}

@article{Hock:2026qrp,
    author = "Hock, Alexander and Marchal, Olivier and Orantin, Nicolas",
    title = "{Geometry of Logarithmic Topological Recursion: Dilaton Equations, Free Energies and Variational Formulas}",
    eprint = "2604.25622",
    archivePrefix = "arXiv",
    primaryClass = "math-ph",
    month = "4",
    year = "2026"
}

@article{Fang2019,
  title = {{On the remodeling conjecture for toric Calabi-Yau 3-orbifolds}},
  volume = {33},
  ISSN = {0894-0347},
  url = {http://dx.doi.org/10.1090/jams/934},
  DOI = {10.1090/jams/934},
  number = {1},
  journal = {J. Amer. Math. Soc},
  publisher = {American Mathematical Society (AMS)},
  author = {Fang,  B. and Liu,  C. and Zong,  Z.},
  year = {2019},
  pages = {135–222}
}

@article{Krichever:1992qe,
    author = "Krichever, I. M.",
    editor = "Khalatnikov, I. M. and Mineev, V. P.",
    title = "{The tau function of the universal Whitham hierarchy, matrix models and topological field theories}",
    eprint = "hep-th/9205110",
    archivePrefix = "arXiv",
    reportNumber = "LPTENS-92-18",
    journal = "Commun. Pure Appl. Math.",
    volume = "47",
    pages = "437",
    year = "1994"
}

@article{Bouchard:2023yau,
    author = "Bouchard, Vincent and Kramer, Reinier and Weller, Quinten",
    title = "{Topological recursion on transalgebraic spectral curves and Atlantes Hurwitz numbers}",
    eprint = "2304.07433",
    archivePrefix = "arXiv",
    primaryClass = "math-ph",
    doi = "10.1016/j.geomphys.2024.105306",
    journal = "J. Geom. Phys.",
    volume = "206",
    pages = "105306",
    year = "2024"
}

@article{Chekhov2019,
  title = {Topological recursion with hard edges},
  volume = {30},
  ISSN = {1793-6519},
  url = {http://dx.doi.org/10.1142/S0129167X19500149},
  DOI = {10.1142/s0129167x19500149},
  number = {03},
  journal = {Internat. J. Math.},
  publisher = {World Scientific Pub Co Pte Ltd},
  author = {Chekhov,  L. and Norbury,  P.},
  year = {2019},
  pages = {1950014}
}

@article{Norbury2014,
  title = {{Gromov–Witten invariants of $\mathbb{P}^1$ and Eynard–Orantin invariants}},
  volume = {18},
  ISSN = {1465-3060},
  url = {http://dx.doi.org/10.2140/gt.2014.18.1865},
  DOI = {10.2140/gt.2014.18.1865},
  number = {4},
  journal = {Geom. Top.},
  publisher = {Mathematical Sciences Publishers},
  author = {Norbury,  P. and Scott,  N.},
  year = {2014},
  pages = {1865–1910}
}

@article{Eynard2011Hurwitz,
  title = {{The Laplace Transform of the Cut-and-Join Equation and the Bouchard–Mariño Conjecture on Hurwitz Numbers}},
  volume = {47},
  ISSN = {1663-4926},
  url = {http://dx.doi.org/10.2977/PRIMS/47},
  DOI = {10.2977/prims/47},
  number = {2},
  journal = {Publ. Res. Inst. Math. Sci.},
  publisher = {European Mathematical Society - EMS - Publishing House GmbH},
  author = {Eynard,  B. and Mulase,  M. and Safnuk,  B.},
  year = {2011},
  pages = {629–670}
}

@article{Bertola:2005he,
    author = "Bertola, M.",
    title = "{Two-matrix model with semiclassical potentials and extended Whitham hierarchy}",
    eprint = "hep-th/0511295",
    archivePrefix = "arXiv",
    doi = "10.1088/0305-4470/39/28/S05",
    journal = "J. Phys. A",
    volume = "39",
    pages = "8823--8856",
    year = "2006"
}

@incollection {Bouchard2008,
    AUTHOR = {Bouchard, Vincent and Mari\~no, Marcos},
     TITLE = {Hurwitz numbers, matrix models and enumerative geometry},
 BOOKTITLE = {From {H}odge theory to integrability and {TQFT} tt*-geometry},
    SERIES = {Proc. Sympos. Pure Math.},
    VOLUME = {78},
     PAGES = {263--283},
 PUBLISHER = {Amer. Math. Soc., Providence, RI},
      YEAR = {2008},
      ISBN = {978-0-8218-4430-4},
   MRCLASS = {14N10 (14N35 81T30 81T45)},
  MRNUMBER = {2483754},
MRREVIEWER = {Yunfeng\ Jiang},
       DOI = {10.1090/pspum/078/2483754},
       URL = {https://doi.org/10.1090/pspum/078/2483754},
}

@article{DuninBarkowski2014,
  title = {{Identification of the Givental Formula with the Spectral Curve Topological Recursion Procedure}},
  volume = {328},
  ISSN = {1432-0916},
  url = {http://dx.doi.org/10.1007/s00220-014-1887-2},
  DOI = {10.1007/s00220-014-1887-2},
  number = {2},
  journal = {Comm Math. Phys.},
  publisher = {Springer Science and Business Media LLC},
  author = {Dunin-Barkowski,  P. and Orantin,  N. and Shadrin,  S. and Spitz,  L.},
  year = {2014},
  pages = {669–700}
}

@article{KPandTaufunctions,
author = {Bychkov, B. and Dunin-Barkowski, P. and Kazarian, M. and Shadrin, S.},
title = {{Topological recursion for Kadomtsev–Petviashvili tau functions of hypergeometric type}},
journal = {J. London Math. Soc.},
volume = {109},
number = {6},
pages = {e12946},
doi = {https://doi.org/10.1112/jlms.12946},
year = {2024}
}

\end{document}